\documentclass[12pt]{article}
\usepackage{amsmath, amsthm, amsfonts, amssymb, amsxtra}
\usepackage{graphicx,psfrag,epsf}
\usepackage{enumerate}
\usepackage{natbib}
\usepackage{url} % not crucial - just used below for the URL
\usepackage{bm, bbm}
\usepackage{booktabs}
\usepackage[pagebackref=false]{hyperref} % pagebackref to true to check how many
\usepackage{xcolor}
\usepackage{subcaption}
\usepackage{tikz}
\usepackage{physics} % for \dd differentiation operator
\usepackage[]{lineno}
\usepackage[margin = 1in]{geometry}
\usepackage{setspace}
\usepackage{dsfont}
\usepackage[plain,noend]{algorithm2e}

\newcommand{\class}[1]{\mathcal{#1}}

\def\ds{\displaystyle}

\newtheorem{theorem}{Theorem}[section]

\theoremstyle{definition}

\graphicspath{ {img/}}
\hypersetup{
    colorlinks = true,
    linkcolor  = blue,
    urlcolor   = blue,
    citecolor  = blue,
    pdftitle   = {Bridging Theory and Practice} 
  }

\begin{document}

\title{Bridging Theory and Practice in Efficient Gaussian Process-Based Statistical Modeling for Large Datasets}

\author{
  Flavio B. Gon\c{c}alves\thanks{Department of Statistics, Universidade Federal de Minas Gerais, Brazil} \and
  Marcos O. Prates\footnotemark[1] \and
  Gareth O. Roberts\thanks{Department of Statistics, University of Warwick, Coventry, UK}
}

\maketitle
\begin{abstract}

Geostatistics is a branch of statistics concerned with stochastic processes over 
continuous domains, with Gaussian processes (GPs) providing a flexible and principled 
modelling framework. However, the high computational cost of simulating or computing 
likelihoods with GPs limits their scalability to large datasets. This paper introduces 
the piecewise continuous Gaussian process (PCGP), a new 
process that retains the rich probabilistic structure of traditional GPs while offering 
substantial computational efficiency. As will be shown and discussed, existing scalable 
approaches that define stochastic processes on continuous domains  --such as the 
nearest-neighbour GP (NNGP) and the radial-neighbour GP (RNGP)--  rely on conditional 
independence structures that effectively constrain the measurable space on which the 
processes are defined, which may induce undesirable probabilistic behaviour and 
compromise their practical applicability, particularly in complex latent GP models. 
The PCGP mitigates these limitations and provides a theoretically grounded and 
computationally efficient alternative, as demonstrated through numerical illustrations.

\end{abstract}
\doublespacing{}
\section{Introduction}

Geostatistics is a prominent and continuously evolving field within spatial statistics. It involves probabilistic models based on stochastic processes defined over continuous domains, typically compact subsets of $\mathbb{R}^2$. A substantial portion of the geostatistical literature is built upon Gaussian processes (GPs), whose favourable probabilistic properties --such as well-defined sample paths, closed-form finite-dimensional distributions and coherent predictive behaviour-- make them powerful and principled tools for statistical modelling. As a result, GPs have become a cornerstone of modern geostatistical methods.

However, the main drawback of traditional GPs lies in their high computational cost, particularly when data are abundant or spatial resolution is high. This has motivated a broad effort to develop computationally efficient approximations and scalable inference methods. Crucially, these solutions should maintain as much as possible of the probabilistic structure of the original GP so that robust statistical analysis remains feasible. This requirement is particularly important when inference involves path-dependent functionals, hierarchical latent GP models, or complex spatial priors. Yet, designing scalable processes that retain the stability of GPs is challenging: natural strategies for reducing computational cost often rely on strong conditional independence assumptions or discretisation schemes that can inadvertently undermine key probabilistic properties.

In practice, several scalable GP approximations provide reasonable predictions at finitely many locations, but may struggle to define coherent stochastic processes over continuous domains or to support inference on functionals involving the entire sample path. In particular, some approaches exhibit probabilistic behaviour that is highly sensitive to data resolution and may lead to unstable or arbitrary conclusions when extrapolated to integrals, maxima or other global functionals. These limitations are especially pronounced in latent GP models, where a well-defined posterior distribution over the entire process is required.

These considerations motivate the search for new stochastic process constructions that simultaneously (i) reduce computational cost and (ii) preserve the continuous-domain probabilistic structure required for coherent inference. The present paper introduces such a process: the \emph{piecewise continuous Gaussian process (PCGP)}. It retains the local conditional structure that enables efficient computation, but embodies a piecewise continuous formulation that preserves essential properties of traditional GPs. As shown in Section~\ref{sec.:PCGP}, this construction yields valid stochastic processes on continuous domains, supports rigorous inference on global functionals and, importantly, ensures methodological validity in latent GP models by defining proper posterior distributions. To the best of our knowledge, the PCGP is the first process that achieves a well-balanced compromise between computational scalability and probabilistic coherence.

The PCGP relates to a broader class of scalable GP approximations, including low-rank methods \citep[e.g.,][]{quinonero2005unifying,banerjee2008gaussian,cressie2008fixed,negahban2011estimation,banerjee2013efficient,shah2019low}, Vecchia-type constructions \citep[e.g.,][]{Vecchia88,NNGP,katzfuss2017multi,katzfuss2021general,BNNGP2,BNNGP1,RNGP}, and stochastic partial differential equations \citep{lindgren2011explicit}. Among these, the nearest-neighbour GP \citep[NNGP,][]{NNGP}, blockNNGP \citep{BNNGP2,BNNGP1}, and radial-neighbour GP \citep[RNGP,][]{RNGP} are notable in that they do define valid stochastic processes over continuous domains, unlike purely finite-dimensional specifications. Nevertheless, their probabilistic structure is intrinsically tied to the chosen directed neighbourhood construction and to the resolution of the data, which may lead to undesirable behaviour when inference depends on functionals of the entire process. In settings such as latent GP models --where a posterior distribution over the full process must be well-defined-- these sensitivities can undermine the validity or robustness of statistical inference.

In contrast, the PCGP provides a process-based construction that preserves essential properties of traditional GPs while remaining computationally efficient. This enables its direct use in models requiring well-defined latent processes or posterior measures, without modifying the inferential goals or restricting the class of functionals under consideration.

The remainder of the paper is organised as follows. Section~\ref{sec.:NNGP} discusses possible approaches to specify random fields in spatial models and the statistical consequences of those that do not define an stochastic process or that lack desirable probabilistic properties. Section~\ref{sec.:PCGP} introduces the PCGP in detail and presents its theoretical properties, while numerical illustrations are given in Section~\ref{subs_NENNGP}. Concluding remarks are provided in Section~\ref{s:conc}, and additional probability theory and results are presented in Appendix.

\section{Gaussian process choices in geostatistical modelling}\label{sec.:NNGP}

Let $\mathcal{D} \subset \mathbb{R}^d$ be a continuous domain, typically a connected open subset, on which we define a Gaussian process
$X = \{X_u : u \in \mathcal{D}\}$ with mean zero and continuous and positive definite covariance function $C(u_1,u_2)$, $u_1,u_2 \in \mathcal{D}$.
Throughout, we refer to such an $X$ as the \emph{parent} Gaussian process.

Gaussian processes provide a flexible framework for constructing spatial models across a wide range of data types.
We briefly review a few common settings.

\begin{enumerate}
\item \textbf{Spatially sampled data.}
In the most direct case, $X$ is observed (up to error) at a finite set of pre-specified locations
$T = \{u_1,\ldots,u_k\}$.
For instance, let $Y_i$ follow an exponential distribution such that, for a non-negative link function $g$,
\[
Y_i \sim \mathrm{Exp}(g(X_{u_i})), \qquad 1 \le i \le k,
\]
leading to the conditional likelihood
\[
L = \prod_{i=1}^k g(X_{u_i}) \exp\{-g(X_{u_i}) y_i\}.
\]

\item \textbf{Point pattern data.}
When $X$ is used to describe spatial variation in hazard or intensity, a common choice is a conditionally heterogeneous
Poisson process (PP) model for a spatial point pattern
$Y = \{u_1,\ldots,u_{k(Y)}\}$,
\[
Y \mid X \sim \mathrm{PP}(g(X)),
\]
where $g : \mathbb{R} \to \mathbb{R}^+$ is a link function acting pointwise on elements of $\mathbb{R}^\mathcal{D}$.
A popular example is the log-Gaussian Cox process (LGCP), obtained by setting $g(x) = \exp(x)$.
The corresponding conditional likelihood is
\[
L = \exp\left\{- \int_{\mathcal{D}} g(X_u)\, du \right\}
\prod_{i=1}^{k(Y)} g(X_{u_i}).
\]

\item \textbf{Areal data.}
Given a partition of $\mathcal{D}$ into disjoint regions $\{\mathcal{D}_i : 1 \le i \le l\}$,
one may aggregate point process counts over each region and adopt the model
\[
Y_i \sim \mathrm{Poisson}\!\left( \int_{\mathcal{D}_i} g(X_u)\, du \right).
\]
The resulting conditional likelihood is
\[
L = \prod_{i=1}^l
\frac{\exp\!\left\{-\int_{\mathcal{D}_i} g(X_u)\, du\right\}
\left(\int_{\mathcal{D}_i} g(X_u)\, du\right)^{y_i}}{y_i!}.
\]
\end{enumerate}

In all cases, the latent GP framework allows prediction of $X$ both at observed and unobserved locations.
The key distinction between spatially sampled data and the point pattern or areal data settings
is that the latter two explicitly involve integrals of the latent process.
This requires the sample paths of the Gaussian process to be measurable functions.
For the most commonly used classes of Gaussian processes, such as those with Mat\'{e}rn and related covariance functions,
this property holds for almost all sample paths
\citep[see, for example, Chapter~4 of][]{Rasmussen2006Gaussian}.

When $k$ is large and the covariance function $C$ is dense, as in these examples,
inference and prediction for $X$ incur a computational cost of order $\mathcal{O}(k^3)$,
which is prohibitive in large-scale problems.
This motivates the use of approximations to $X$.
Let $Z$ denote such an approximation.

A \emph{principled} approximation requires that the collection of finite-dimensional distributions of $Z$
satisfy Kolmogorov consistency \citep[see, for example, Section~36 of][]{Billin}:
for any finite sets $D_1 \subset D_2 \subset \mathcal{D}$, the law of $Z_{D_2}$ projected onto $D_1$
must coincide with the law of $Z_{D_1}$ defined directly.
Kolmogorov consistency guarantees the existence of $Z$ as a stochastic process
\citep[][Theorem~36.1]{Billin}.
It does not, however, ensure regularity properties of $Z$.
Even when the marginal distributions of $Z_u$ provide good approximations to those of $X_u$
for each fixed $u \in \mathcal{D}$, the process $Z$ may fail to approximate more complex functionals, such as integrals of $X$ over $\mathcal{D}$.

Several approximations proposed in the literature fail to satisfy the measurability requirement,
and some do not even define valid stochastic processes.
Vecchia-based constructions such as the nearest neighbour Gaussian process (NNGP; \citealp{NNGP})
and the radial neighbour Gaussian process (RNGP; \citealp{RNGP})
do satisfy Kolmogorov consistency, but do not guarantee almost sure measurability of the sample paths as functions on the continuous spatial domain $\mathcal{D}$.

In many situations, such approximations are adequate for statistical purposes.
This is typically the case for spatially sampled data, where inference and prediction concern
only finite collections of locations $T_{\mathrm{new}}$.
More complex functionals, such as integrals or suprema, may fail to be well defined under
the law of $Z$ and, in some cases, may not even be measurable under its construction.
For models such as the LGCP for point pattern data or GP-based models for areal data,
the likelihood itself may therefore fail to be well defined, rendering any formal likelihood-based approach impossible.

A common workaround is to introduce discrete spatial approximations of $Z$,
for example, piecewise constant processes induced by multivariate Normal distributions on a grid.
While widespread in practice, such approaches introduce additional approximation bias and,
as the discretisation is refined, often lead to substantial deterioration in Markov chain Monte Carlo performance.
This behaviour reflects the fact that the underlying infinite-dimensional model may fail to exist
as a proper probability measure, and consequently its posterior distribution may not be well defined.

These considerations motivate the search for new Gaussian process constructions --such as the one
introduced in this paper-- that define valid stochastic processes on continuous domains with rich measurability
properties, while remaining computationally efficient and scalable.
Our proposed model, the piecewise continuous Gaussian process, is described in
Section~\ref{sec.:PCGP}.

The remainder of this section focuses on GP models that are consistent but lack sufficient regularity,
using the NNGP as a central example.
The discussion extends directly to related constructions such as the block-NNGP and the RNGP.
We begin with an overview of the NNGP and then examine its statistical limitations, highlighting
the settings in which such models may or may not be appropriate for inference.

\subsection{The Nearest Neighbour Gaussian Process (NNGP)}\label{NNGP}

Gaussian processes are a fundamental tool in spatial statistics, offering flexible modelling and strong probabilistic foundations. However, their computational demands --particularly the $\mathcal{O}(k^3)$ cost of Cholesky decompositions for size-$k$ matrices-- make them infeasible for large datasets.

To overcome this, \citet{Vecchia88} introduced a scalable approximation based on conditional independence, known as the Vecchia approximation. The key contribution of \citet{NNGP} was to extend this idea to the infinite-dimensional setting, thereby defining a stochastic process over continuous domains. This led to the NNGP, constructed via directed acyclic graphs (DAGs) and rigorously shown to satisfy Kolmogorov's existence theorem. Prior to their work, Vecchia approximations had been applied only to a finite set of observed and predicted locations, without considering whether an underlying infinite-dimensional process existed. Some of the most recent literature on approximate methods include \citet{katzfuss2017multi}, \citet{zhang2019smoothed}, \citet{katzfuss2021general}, \citet{wu2022variational}, \citet{wu2024large}. Recently, \citet{RNGP} proposed a variation of the NNGP --the RNGP-- by modifying only the definition of the neighouring sets under a specific ordering of the reference set.

Let $\mathcal{S} = \{\mathfrak{s}_1, \ldots, \mathfrak{s}_r\} \subset \mathcal{D}$ be an ordered reference set of spatial locations. The finite-dimensional distributions of the NNGP (and  RNGP) are constructed by factorizing in terms of the marginal distribution of $Z_{\mathcal{S}}$ and the conditional distributions of any finite set of additional locations given $Z_{\mathcal{S}}$. Let $\tilde{\pi}$ and $\pi$ denote the densities under the measures of the NNGP and the parent GP, respectively, where the latter is the process that the NNGP is designed to approximate. Then
\begin{equation}\label{eq_NNGP1}
\tilde{\pi}(Z_{\mathcal{S}}) \;=\; \prod_{i=1}^r \pi\!\left(Z_{\mathfrak{s}_i} \mid Z_{\mathcal{N}(\mathfrak{s}_i)}\right),
\end{equation}
where $\mathcal{N}(\mathfrak{s}_i)$ denotes a subset of $\mathcal{N}_i=\{\mathfrak{s}_1,\ldots,\mathfrak{s}_{i-1}\}$.

In \citet{NNGP}, for a fixed choice of $m \in \mathbb{N}$, the authors 
take $\mathcal{N}(\mathfrak{s}_i)=\mathcal{N}_i$ for $i \le m$, 
and for $i>m$ define $\mathcal{N}(\mathfrak{s}_i)$ to be the set of the 
$m$ locations in $\mathcal{N}_i$ closest to $\mathfrak{s}_i$, that is, 
\[
\mathcal{N}(\mathfrak{s}_i)=
\left\{
\mathfrak{s}_j \in \mathcal{N}_i :
%\|\mathfrak{s}_j - \mathfrak{s}_i\|
d(\mathfrak{s}_j , \mathfrak{s}_i)
\text{ is among the } m \text{ smallest}
\right\},
\]
so that $\lvert \mathcal{N}(\mathfrak{s}_i) \rvert = m$, with $d(\cdot,\cdot)$ being some distance function adopted in the definition of the covariance function of the parent GP. \citet{RNGP}, on the other hand, define $\mathcal{N}(\mathfrak{s}_i)=\left\{
\mathfrak{s}_j \in \mathcal{N}_i :
%\|\mathfrak{s}_j - \mathfrak{s}_i\| 
d(\mathfrak{s}_j , \mathfrak{s}_i) < R\right\}$, for some pre-specified radius $R$.

For any finite set $D_0=\{u_1,\ldots,u_n\}$ in $\mathcal{D}\setminus\mathcal{S}$, the coordinates of $Z_{D_0}$ are conditionally independent given $Z_{\mathcal{S}}$, i.e.
\[
\tilde{\pi}(Z_{D_0} \mid Z_{\mathcal{S}}) \;=\; \prod_{i=1}^n \pi\!\left(Z_{u_i} \mid Z_{\tilde{\mathcal{N}}(u_i)}\right),
\]
where $\tilde{\mathcal{N}}(u_i)$ denotes any subset of the reference set, for example, $
\left\{
\mathfrak{s}_j \in \mathcal{S} :
%\|\mathfrak{s}_j - \mathfrak{s}_i\|
d(\mathfrak{s}_j, \mathfrak{s}_i)
\text{ is among the } m \text{ smallest}
\right\}$ \citep{NNGP}, or $\left\{
\mathfrak{s}_j \in \mathcal{S} :
%\|\mathfrak{s}_j - \mathfrak{s}_i\|
d(\mathfrak{s}_j, \mathfrak{s}_i)
< R\right\}$ \citep{RNGP}. Algorithm 1 describes the procedure for simulating the NNGP at a finite set of locations.

\cite{NNGP} show that this specification of finite-dimensional distributions satisfies the consistency conditions of Kolmogorov's existence theorem, thereby defining a stochastic process on a continuous domain. The resulting process admits a sparse precision structure, which enables fast inference and scalable simulation. Nonetheless, as discussed in the next section, its probabilistic construction imposes important practical limitations.

\begin{figure}[!htbp]
\centering
\fbox{%
\begin{minipage}{0.9\linewidth}
\centering
\emph{Algorithm 1: Simulation from the NNGP}
\vspace{0.5em}

\textbf{Input:} Parent GP mean and covariance functions; domain $\mathcal{D}$; reference set $\mathcal{S}=\{\mathfrak{s}_1,\ldots,\mathfrak{s}_r\}$; neighbour size $m$; target set $D_0\subset \mathcal{D}\setminus\mathcal{S}$.

\vspace{0.3em}
\textbf{Output:} Simulated values $\{Z_u : u\in D_0\}$.

\begin{enumerate}
  \item \textbf{Simulate the reference process:}
For $i=1,\ldots,r$, sample
    $
    Z_{\mathfrak{s}_i}
    \sim
    \pi\!\left(\cdot \mid Z_{\mathcal{N}(\mathfrak{s}_i)}\right).
    $

  \item \textbf{Simulate at $D_0$:}
For each $u_i\in D_0$, independently sample
    $
    Z_{u_i}
    \sim
    \pi\!\left(\cdot \mid Z_{\mathcal{N}(u_i)}\right).
    $
\end{enumerate}

\end{minipage}}
\end{figure}

\subsection{Shortcomings of the NNGP}

The defining feature of the NNGP --its sparse conditional independence structure-- yields significant computational benefits but introduces important statistical limitations. Given the values of $Z$ at the reference set $\mathcal{S}$, the NNGP construction implies that the process is conditionally independent at all remaining locations in $\mathcal{D} \setminus \mathcal{S}$. Consequently, although the NNGP defines a stochastic process, its law is supported on a measurable space that lacks several of the structural properties enjoyed by the canonical path space of Gaussian processes.

Such construction implies important practical limitations: 

\begin{itemize}
  \item when the process is evaluated at high resolution (e.g., tens or hundreds of thousands of locations), the NNGP exhibits bias and a loss of robustness in inference. For example, in numerical experiments (presented in detail in Section~\ref{subs_NENNGP}) maximum likelihood estimates of the process mean are considerably biased as the sample size increases (see Table \ref{tab3});

  \item when spatial phenomena are expected to vary smoothly. For example, applications that involve extreme value analysis --such as modeling temperature or precipitation maxima in climate studies \citep{sang2010continuous}-- the NNGP may suffer severe losses in statistical robustness. Results in Table~\ref{tab1} show that empirical estimates of the minimum and maximum of the path do not stabilise as the sample size increases, getting arbitrarily extreme;
  
  \item when the model explicitly involves global functionals of the path, e.g., spatial fusion models \citep{FSNNGP1, FSNNGP2}, spatial point processes \citep{PPNNGP1, PPNNGP2, PPNNGP3, gonccalves2020exact}, or spatial quantile regression \citep{QRNNGP1}, these quantities, defined by the NNGP, are ill-defined and may result in statistical instability. Results in Table \ref{tab1} show that estimates of the variance of integrals of the path, that should stabilise, get arbitrarily small when increasing the sample size, being approximately 15/60/275 times smaller (than the true value) for sample sizes of 10/40/160 thousand.%, see Section~\ref{s:nngp_short}. 

  \item in standard inferential procedures such as maximum likelihood and Bayesian estimation. For example, \citet{Michelin2025fast} observed that in practical applications, such as porosity estimation in petroleum reservoirs, the NNGP underestimates posterior variances. This occurs because the conditional independence assumptions reduce the effective dependence among observations, making the data appear more informative than they truly are. %A similar effect should be expected for the variance of maximum likelihood estimators, implying that both credibility and confidence intervals may be too narrow.
\end{itemize}

Although the irregularity of NNGP paths and the limitations in the measurability of certain functionals can compromise both the theoretical validity and robustness of statistical inference, the NNGPs often perform well in practice, especially when the process is required to be unveiled at a moderate (typically $\class{O}(10^3)$) number of locations. In such settings, one may interpret the model pragmatically: the sparse conditional structure is treated as a computational device, while the underlying process is implicitly assumed to behave like the parent GP elsewhere. This interpretation explains the empirical success of NNGP-based models in many applications.

In the next section, we present the theoretical grounds for the practical pathologies previously described.

\subsection{Theoretical foundation of the NNGP shortcomings}
\label{s:nngp_short}
%The remainder of t
This section delves deeper into the statistical shortcomings of the NNGP by discussing three results that formalise some of the theoretical pathologies of the NNGP (proofs can be found in Appendix B).

In what follows, let $Z$ be an NNGP on a compact region $\mathcal{D} \subset \mathbb{R}^d$, defined on the probability space \((\mathbb{R}^\mathcal{D}, \mathcal{R}^\mathcal{D}, P)\). A formal definition of this probability space along with the measure theory behind the definition of the NNGP are provided in Appendix A. Traditional GPs, on the other hand, are naturally defined on $(C,\mathcal{C})$, where $C$ denotes the Banach space of continuous real-valued functions on $\mathcal{D}$ equipped with the supremum norm, and $\mathcal{C}$ is the corresponding Borel $\sigma$-algebra.

The key structural feature of the NNGP and of other sparse GP formulations that leads to the undesirable properties discussed below is the conditional independence of the infinite-dimensional remainder of the process, given $Z_{\mathcal{S}}$.

\begin{theorem}
\label{RPT}
\textbf{Highly irregular paths.}
Let $D_0$ be any countable set dense in $\mathcal{D}\setminus\mathcal{S}$, and let $\mathfrak D$ denote the collection of open hypercubes contained in $\mathcal D$ with rational endpoints. Then, with probability one, for every $\mathcal D^*\in\mathfrak D$ the set $\{Z_u:u\in D_0\cap \mathcal D^*\}$ is dense in $\mathbb R$.
\end{theorem}

\begin{theorem}
\label{mMT}
\textbf{Unbounded paths.}
Let $D_n = \{u_1, \ldots, u_{r_n}\}$ be a regular grid of points in $\mathcal{D}$, where $D_n=D_{n-1}\cup \bar{D}_n$ and $D_n$ doubles the precision of $D_{n-1}$. Then:
\begin{enumerate}[(i)]
\item For any $x, y \in \mathbb{R}$, there exist countably infinite sets $D_L, D_U \subset \mathcal{D}$ such that, $P$-a.s., $Z_u < x$ for all $u \in D_L$ and $Z_u > y$ for all $u \in D_U$.
\item $\ds \lim_{n\rightarrow\infty}\min_{u\in D_n}{Z_{u}} = -\infty$ and $\ds \lim_{n\rightarrow\infty}\max_{u\in D_n}{Z_{u}} = \infty$, $P$-a.s.
\end{enumerate}
\end{theorem}

The result in Theorem~\ref{RPT} raises serious concerns regarding the use of the NNGP to model %real-world spatial phenomena that are expected to vary smoothly. 
smooth spatial fields. In particular, it suggests that the sample paths of the NNGP are too irregular to adequately represent such processes. Theorem \ref{mMT} further shows that estimators of the pathwise maximum and minimum under the NNGP become increasingly biased as the resolution of the spatial grid increases. This bias grows unboundedly with the number of evaluation points. %As a consequence, applications involving extreme value analysis -- such as modeling temperature or precipitation maxima in climate studies \citep{sang2010continuous} -- may suffer from severe losses in statistical robustness.

Another critical property of the NNGP is the non-measurability of functionals that depend on the process evaluated over uncountable subsets of the domain such as path integrals and level sets. As a result, statistical models or analyses involving such quantities are, in a strict sense, mathematically invalid under the NNGP framework. In practical applications, this can lead to unstable or unreliable estimates. For example, approximating integrals of the type $\int_{\mathcal{D}}g(Z_u)du$ via partial sums yields limits that depend solely on the reference set and the number of neighbours. This underestimation of randomness undermines statistical robustness and interpretability. The following theorem formalises this issue.

\begin{theorem}
\label{CPS}
\textbf{Convergence of partial sums.}
Let $D_n = \{u_1, \ldots, u_{r_n}\}$ be a regular grid as defined in Theorem~\ref{mMT}. Assume that, for each $u_i\in D_n$, the conditional distribution of $Z_{u_i}$ is defined using the $m$ closest neighbors of $u_i$ in $\mathcal S$, with a deterministic tie-breaking rule. Then:
\begin{enumerate}[(i)]
\item
\[
\lim_{n \to \infty} \sum_{i=1}^{r_n} \frac{1}{r_n} Z_{u_i}
=
\left( \sum_{j=1}^{r} \left[ \lim_{n \to \infty} \sum_{i=1}^{r_n} \frac{a_{ij}}{r_n} \right] Z_{\mathfrak{s}_j} \right),
\qquad P\text{-a.s.},
\]
where, for each $j=1,\ldots,r$,
$\ds \left(\lim_{n \to \infty} \sum_{i=1}^{r_n} \frac{a_{ij}}{r_n}\right)
$ exists.

\item For any measurable function $g:\mathbb R\to\mathbb R$ such that $g(\mu+\sigma W)$ has finite mean and variance for all $(\mu,\sigma)\in\mathbb R\times\mathbb R^+$ and $W\sim N(0,1)$,
\[
\lim_{n \to \infty} \sum_{i=1}^{r_n} \frac{1}{r_n} g(Z_{u_i})
=
\lim_{n \to \infty} \sum_{i=1}^{r_n} \frac{\mu_{g,i}}{r_n},
\qquad P\text{-a.s.},
\]
where
\[
\mu_{g,i}
=
\mathbb E_W\!\left[
g\!\left(\sum_{j=1}^{r} a_{ij} Z_{\mathfrak{s}_j} + \sigma_i W\right)
\right],
\]
$\sigma_i^2=\mathrm{Var}(Z_{u_i}\mid Z_{\mathfrak{s}_1},\ldots,Z_{\mathfrak{s}_r})$, and $\ds
\left(\lim_{n \to \infty} \sum_{i=1}^{r_n} \frac{\mu_{g,i}}{r_n}\right)$ 
exists.
\end{enumerate}
\end{theorem}

This theorem highlights a troubling aspect: although integrals like $\int_{\mathcal{D}} g(Z_u) du$ are not well-defined under the NNGP (due to non-measurability), their discrete approximations converge to quantities that depend only on $Z_{\mathcal{S}}$ and $m$. This dependency undermines the robustness of the NNGP by introducing sensitivity to both the choice of the reference set  and the maximum number of neighbours. For instance, setting $g(Z_u) = \mathbbm{1}_A(Z_u)$, where $\mathbbm{1}_A$ is the indicator function over a set $A\subset \mathbb{R}$, gives a discrete approximation for the volume of the (non-measurable) level set $\mathcal{D}_A(Z)=\{u \in \mathcal{D} : Z_u \in A\}$ -- a quantity that is not well-defined under the NNGP but commonly used in spatial statistics.

% Moreover, the sparse conditional structure of the NNGP can affect standard procedures such as maximum likelihood and Bayesian estimation. For example,
% \citet{Michelin2025fast} showed that in practical applications, such as porosity estimation in petroleum reservoirs, the NNGP underestimates posterior variances. This occurs because the conditional independence assumptions reduce the effective dependence among observations, making the data appear more informative than they truly are. A similar effect should be expected for the variance of maximum likelihood estimators, implying that both credibility and confidence intervals may be too narrow.

Alternative models based on modified NNGP structures or partitioning schemes 
\citep{BNNGP1, BNNGP2, zheng2021nearest, zheng2023nearest} are subject to analogous theoretical concerns. 
The RNGP differs from the NNGP only through the construction of the neighbour sets $\mathcal{N}(u)$, yet it also imposes conditional independence on the infinite-dimensional remainder given $X_{\mathcal{S}}$, and thus inherits the same foundational limitations.

The same assumption is made in the ProSpar-GP process \citep{LiMak2025ProSparGP}, which likewise induces the same issues. 
While ProSpar-GP satisfies Kolmogorov's existence conditions and is specifically designed for massive nonstationary datasets via a combination of sparse GP components, its construction still relies on imposing conditional independence of the infinite-dimensional remainder after specifying the parent GP finite-dimensional distribution at a set of inducing points.

Alternative covariance approximation strategies --such as low-rank methods, process convolutions, and various Vecchia-based approaches-- are widely used in the spatial statistics literature due to their computational efficiency and ability to produce smoother sample paths. However, these methods typically define approximations only at a finite set of locations and do not induce infinite-dimensional stochastic processes. This limits their theoretical coherence and can hinder rigorous inference for path-level quantities. Low-rank approaches, in particular, often require repeated matrix decompositions or recomputations during inference, complicating the assessment and control of approximation error. Additionally, many Vecchia-based techniques rely on a fixed ordering of spatial locations, which makes them incompatible with retrospective sampling schemes that are essential for exact inference in infinite-dimensional models.

\subsection{Didactic example: the Brownian bridge} \label{ssec.:SBB}

We start illustrating the limitations of the NNGP using a simple one-dimensional example: the standard Brownian bridge (BB), which is a Brownian motion on $[0,1]$ conditioned to start and end at zero. For the BB, key path-level functionals such as the integral and maximum have known marginal distributions.

Let $Z \sim \text{BB}(0,1)$. Then:
\begin{itemize}
\item The path integral $\int_{0}^1Z_udu \sim N(0, 1/12)$.
\item The maximum has density $\pi(x) = 4x e^{-2x^2}, x > 0$, with mean $\sqrt{\pi/8}$ and variance $(4-\pi)/8$.
\end{itemize}

Under the NNGP measure induced by the $\text{BB}(0,1)$, the distribution of $Z$ at the reference set remains the same as in the parent process due to the Markov property. Without loss of generality, let the reference set $\mathcal{S} = \{\mathfrak{s}_1, \ldots, \mathfrak{s}_r\}$ be a regular grid of $r$ points in the interval $(0,1)$, with $r$ even. For each $n \in \mathbb{N}$, define $D_n$ as the set of points forming a regular grid with $2^n$ equally spaced locations within each subinterval $(\mathfrak{s}_j, \mathfrak{s}_{j+1})$. Consider the inequality:
$$
\bar{Z}_m := \max_{u \in D_n} Z_u \geq \max_{u \in D_n^*} Z_u =: \bar{\bar{Z}}_m,
$$
where $D_n^* = \left\{ u \in D_n : \mathfrak{s}_{(r-1)/2} + \epsilon/4 \leq u \leq \mathfrak{s}_{(r+1)/2} - \epsilon/4 \right\}$ and $\epsilon = (r+1)^{-1}$.

Under the NNGP measure, and conditional on $Z_{\mathfrak{s}_{(r-1)/2}}$ and $Z_{\mathfrak{s}_{(r+1)/2}}$, the variable $\bar{\bar{Z}}_m$ is the maximum of independent Gaussian variables. Each of these has a mean given by a convex combination of $Z_{\mathfrak{s}_{(r-1)/2}}$ and $Z_{\mathfrak{s}_{(r+1)/2}}$, with weights bounded away from zero and one. Specifically, the smaller weight lies in the interval $(\epsilon/4, \epsilon/2)$. The variances of these variables lie in the range $(3\epsilon/16, \epsilon/4)$.
It follows directly that:
$$
\lim_{n \rightarrow \infty} \bar{Z}_m \geq \lim_{n \rightarrow \infty} \bar{\bar{Z}}_m = \infty.
$$

Figure \ref{fig_BB} displays the empirical distributions of the path integral and the maximum of a standard Brownian bridge under the NNGP approximation with $r = 5$, across different values of $n$ and based on a sample size of 10,000. On the left-hand side, we observe that while the approximation of the path integral distribution improves with increasing $n$, it continues to underestimate the variance. On the right-hand side, the distribution of the maximum deviates further from the true distribution as $n$ increases. These results confirm that the NNGP does not preserve the distributional properties of the Brownian bridge, particularly with respect to the maximum, where the approximation deteriorates more markedly.

\begin{figure}[!h]
\centering
\includegraphics[height=.25\textheight, width=1\linewidth]{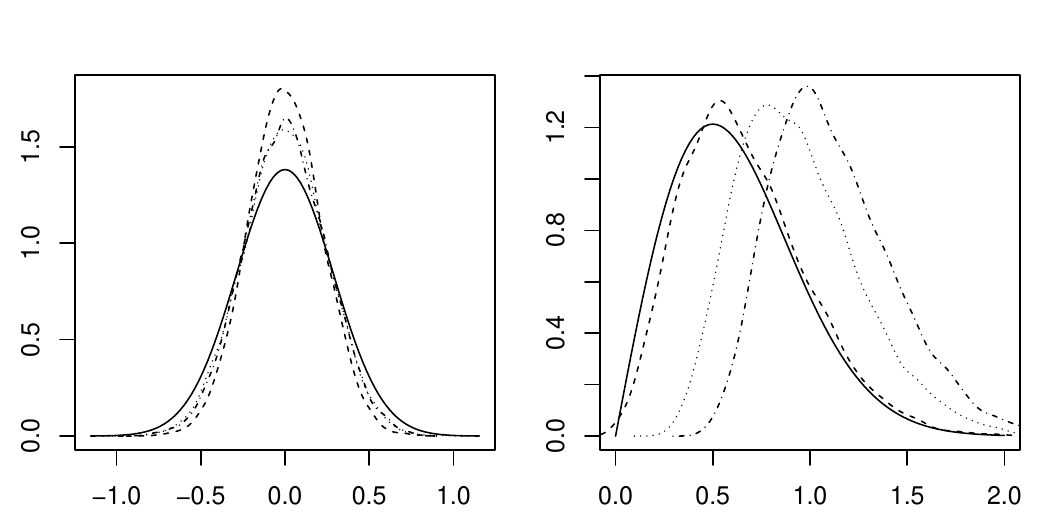}
\caption{Exact (solid) and NNGP empirical densities for the path integral (left) and maximum (right). Values of $n$ are $3$ (dashed), $9$ (dotted), and $14$ (dot-dashed).}
\label{fig_BB}
\end{figure}

A more complex two-dimensional example that illustrates similar issues is presented in Section \ref{subs_NENNGP}, along with the results of the proposed PCGP introduced in the next section.

\section{The piecewise continuous GP}\label{sec.:PCGP}

To overcome the statistical deficiencies of sparse GPs such as the NNGP, while maintaining computational efficiency, we propose a novel sparse Gaussian process construction that preserves the key probabilistic properties of a full GP to the greatest possible extent. This new process, which we call the piecewise continuous GP (PCGP), ensures piecewise continuity and maintains all the measurability properties of traditional GPs. Consequently, it can be effectively used in geostatistical models that rely on global path functionals, including spatial fusion, spatial quantile regression, and point process models.

The PCGP is defined hierarchically and exploits conditional independence structures similarly to sparse Gaussian process approximations. 
Unlike NNGP- or RNGP-type constructions, however, conditional independence is imposed across regions of positive volume rather than at the level of individual locations.

Let $X=\{X_u:u\in \mathcal{D}\}$ be a (parent) Gaussian process on a domain $\mathcal{D}\subset\mathbb{R}^d$ with probability measure $P$. 
Partition $\mathcal{D}$ into $K$ disjoint congruent convex polytopes $\{\mathcal{D}_1,\ldots,\mathcal{D}_K\}$ such that
\[
\mathcal{D}=\bigcup_{k=1}^K \mathcal{D}_k, \qquad \mathcal{D}_k^\circ \cap \mathcal{D}_\ell^\circ = \varnothing \ \text{for } k\neq \ell,
\]
where $\mathcal{D}_k^\circ$ is the interior of $\mathcal{D}_k$, and let $\mathcal{S}=\{\mathfrak{s}_1,\ldots,\mathfrak{s}_r\}\subset \mathcal{D}$ denote a finite reference set.

The PCGP is the stochastic process $Z$ with probability measure $\tilde P$ defined by the factorization
\begin{equation*}
\tilde P(dz)
=
\tilde P(dz_{\mathcal{S}})
\otimes
\bigotimes_{k=1}^K
P\!\left(dz_{\mathcal{D}_k}\mid z_{\mathcal{N}(\mathcal{D}_k)}\right),
%\label{eq:pcgp-factorization}
\end{equation*}
where $\mathcal{N}(\mathcal{D}_k)$ is a selected subset of $\mathcal{S}$, and the marginal distribution on the reference set satisfies
\begin{equation}
\frac{d\tilde P}{d\nu^r}(z_{\mathcal{S}})
=
\tilde\pi(z_{\mathcal{S}})
:=
\prod_{i=1}^r
\pi\!\left(z_{\mathfrak{s}_i}\mid z_{\mathcal{N}(\mathfrak{s}_i)}\right).
\label{eq:pcgp-reference}
\end{equation}
Here, $\nu$ denotes Lebesgue measure on $\mathbb{R}$, $\mathcal{N}(\mathfrak{s}_i)$ is any user-specified subset of $\{\mathfrak{s}_1,\ldots,\mathfrak{s}_{i-1}\}$ and $\pi$ represents the respective density under the parent GP measure.
%The conditional distributions in \eqref{eq:pcgp-reference} are those induced by the parent Gaussian process.
Choosing $\mathcal{N}(\mathfrak{s}_i)=\{\mathfrak{s}_1,\ldots,\mathfrak{s}_{i-1}\}$ recovers the exact GP marginal on $\mathcal{S}$, while sparse choices (as defined in the NNGP and RNGP) yield computational savings.

Conditional on $Z_{\mathcal{S}}$, the restrictions $\{Z_{\mathcal{D}_k}\}_{k=1}^K$ are mutually independent.
For each region $\mathcal{D}_k$, the conditional distribution $P(dz_{\mathcal{D}_k}\mid z_{\mathcal{N}(\mathcal{D}_k)})$ is the Gaussian process obtained by restricting the parent GP to $\mathcal{D}_k$ and conditioning on a subset $\mathcal{N}(\mathcal{D}_k)$ of $\mathcal{S}$.
Typically, this subset consists of the $\ddot m$ reference locations closest to the centroid of $\mathcal{D}_k$.

Under this construction, sample paths are almost surely continuous within each region $\mathcal{D}_k$, while discontinuities  occur across region boundaries.
Owing to the ordered neighbour structure imposed on $Z_{\mathcal{S}}$ and the fact that, conditional on $Z_{\mathcal{S}}$, the restrictions $\{Z_{\mathcal{D}_k}\}_{k=1}^K$ are specified by the parent Gaussian process, the resulting process is itself a (generally non-stationary) Gaussian process on $S$.
The resulting process is therefore piecewise continuous on $S$.
Boundary points are assigned deterministically to one of the adjacent regions, ensuring a well-defined process on all of $S$.

As an illustration, let $\mathcal{D}=[0,1]^2$ and partition it into a $10\times10$ regular grid, yielding $K=100$ regions.
Let $\mathcal{S}$ be the $25\times25$ regular grid on $S$, and assume an NNGP structure on $\mathcal{S}$ with $m$.
For each region $\mathcal{D}_k$, let $\mathcal{N}(\mathcal{D}_k)$ denote the $\ddot m$ reference locations closest to its centroid.

For any finite set $D_0\subset \mathcal{D}\setminus\mathcal{S}$, the joint density of $Z_{D_0}$ under the PCGP is
\[
\tilde{\pi}(z_{D_0})
=
\int
\prod_{k: D_0\cap \mathcal{D}_k\neq\varnothing}
\pi\!\left(z_{D_0\cap \mathcal{D}_k}\mid z_{\mathcal{N}(\mathcal{D}_k)}\right)
\tilde\pi(z_{\mathcal{S}})
\,dz_{\mathcal{S}}.
\]

Algorithm~2 describes a constructive procedure for simulating finite-dimensional realisations of the PCGP at any prescribed set of locations.

\begin{figure}[!htbp]
\begin{center}
\fbox{%
\begin{minipage}{13cm}
\centering
\emph{Algorithm 2: Simulation from the PCGP}
\vspace{0.5em}

\textbf{Input:} Parent GP covariance function; domain $\mathcal{D}$; partition $\{\mathcal{D}_k\}_{k=1}^K$; reference set $\mathcal{S}=\{\mathfrak{s}_1,\ldots,\mathfrak{s}_r\}$; neighbour sizes $m$ and $\ddot m$; target set $D_0\subset \mathcal{D}\setminus\mathcal{S}$.

\vspace{0.3em}
\textbf{Output:} Simulated values $\{Z_u : u\in D_0\}$.

\begin{enumerate}
  \item \textbf{Simulate the reference process:}
For $i=1,\ldots,r$, sample
    $
    Z_{\mathfrak{s}_i}
    \sim
    \pi\!\left(\cdot \mid Z_{\mathcal{N}(\mathfrak{s}_i)}\right).
    $

  \item \textbf{Simulate within regions:}
  For each region $\mathcal{D}_k$ such that $D_0\cap \mathcal{D}_k\neq\varnothing$, sample
    $
    Z_{D_0\cap \mathcal{D}_k}
    \sim
    \pi\!\left(\cdot \mid Z_{\mathcal{N}(\mathcal{D}_k)}\right).
    $
\end{enumerate}

\end{minipage}}
\end{center}
\end{figure}

We focus on the common setting in which conditional independence is induced through nearest-neighbour structures. Specifically, $\mathcal{N}(\mathfrak{s}_i)$ denotes the set of the $m$ nearest locations to $\mathfrak{s}_i$ among the previously ordered reference locations, while $\mathcal{N}(\mathcal{D}_k)$ consists of the $\ddot m$ reference locations closest to the centroid of region $S_k$. The choice of $\ddot m$ plays an important role in preserving the correlation structure of the parent Gaussian process. For example, it is reasonable to set the number of reference points from $\mathcal{S}$ lying within any given region $S_k$ to be small relative to $\ddot m$, which ensures substantial overlap among neighbouring conditioning sets. In practice, taking $\ddot m = m$ is often sufficient and performs well empirically.

The number of regions $K$ governs the trade-off between computational efficiency and approximation fidelity. A practical guideline is to choose $K$ so that the number of non-reference locations to be simulated within each region $\mathcal{D}_k$ remains moderate, for example on the order of a few hundred.

The piecewise continuity of the PCGP yields both theoretical and practical advantages, alleviating key limitations of existing sparse GP constructions. Notably, in the Brownian bridge setting considered in Section~\ref{ssec.:SBB}, the PCGP recovers the exact parent process as a direct consequence of the Markov property.

Existence of the PCGP on a measurable space as rich as that of the parent GP follows directly from the finite dimensionality of 
$\mathcal{S}$ and the existence of a GP on each polytope, as formalised in the theorem below.

\begin{theorem}[Existence of the PCGP]
\label{Exst_PCGP}
Let $X$ be a Gaussian process on $\mathcal{D}\subset\mathbb{R}^d$ with positive definite covariance function and continuous sample paths, inducing a Borel probability measure on $C(S)$.
Let $\{\mathcal{D}_k\}_{k=1}^K$ be a partition of $\mathcal{D}$, and for each $k$ let $
C_k := C(\mathcal{D}_k)$ denote the Banach space of continuous real-valued functions on $\mathcal{D}_k$ equipped with the supremum norm and its Borel $\sigma$-algebra $\mathcal C_k$.
Then the PCGP construction defines a probability measure on the product measurable space
\[
\Big(\bigotimes_{k=1}^K C_k,\ \bigotimes_{k=1}^K \mathcal C_k\Big).
\]
\end{theorem}
The proof of Theorem~\ref{Exst_PCGP} is presented in Appendix~B.

Suppose the process is to be simulated at $n$ locations uniformly distributed over the domain $\mathcal{D}$, and the user wishes to have approximately $N$ of these locations in each $\mathcal{D}_k$. This implies setting the number of polytopes to $K \approx n / N$.
The computational cost of simulating, for example, an NNGP at $n$ locations is $\mathcal{O}(r m^3 + n m^3)$, where $r$ is the size of the reference set and $m$ is the number of neighbours. The cost for the PCGP is also linear in $n$ and, if no neighboring structure is used, is given by $
\mathcal{O}\left( r^3 + n \cdot \frac{(N + r)^3}{N} \right).
$
If the NNGP neighbouring structure is used, the cost is
$
\mathcal{O}\left( r m^3 + n \cdot \frac{(N + \ddot{m})^3}{N} \right).
$

If the user finds the PCGP paths insufficiently smooth due to discontinuities along the partition boundaries, a mixed PCGP (mPCGP) can be employed. The mPCGP combines multiple PCGPs with different partitions to smooth the boundaries. For example, one partition may have centroids placed at the vertices of another. Let $\tilde{\tilde{P}}$ denote the probability measure of the mPCGP and $\tilde{\tilde{\pi}}$ the induced densities. Formally, for $G$ different partitions, the mPCGP is defined as
\begin{eqnarray*}
\tilde{\tilde{\pi}}(z_{\mathcal{S}}) &=& \tilde{\pi}(z_{\mathcal{S}}), \\[4pt]
\tilde{\tilde{P}}\big(dz^{(1)},\ldots,dz^{(G)} \mid z_{\mathcal{S}}\big)
&=&
\bigotimes_{j=1}^{G}
\tilde{P}_j\big(dz^{(j)} \mid z_{\mathcal{S}}\big), \\[4pt]
Z_u &=& \frac{1}{G} \sum_{j=1}^{G} Z_u^{(j)},
\qquad
u \in \mathcal{D} \setminus \mathcal{S},
\end{eqnarray*}
where $\tilde{P}_j$ denotes the probability measure of the PCGP associated with the $j$-th partition and $Z^{(j)}$ is a realization from $\tilde{P}_j$. 
The covariance function associated with each $\tilde{P}_j$ is scaled by a factor $G$ so that the resulting covariance of the averaged process approximates the covariance function of a single PCGP.

For partitions with comparable sizes and areas, the computational cost of the mPCGP is, at most,
\[
O\left( r^3 + n G \frac{(N + r)^3}{N} \right).
\]

Figure S.1 in Supplementary Material S.1 illustrates examples of partitions with $G = 2$ and $G = 4$, using squares in $R^2$. An example comparing the PCGP to the mPCGP with $G=4$ is presented in the next section.

%These limitations are also present in the ProSpar-GP approach \citep{LiMak2025ProSparGP}, which -- while satisfying Kolmogorov's conditions -- is specifically tailored to massive nonstationary datasets through a combination of sparse GP components. Its sparse structure is constructed by first considering the full GP finite-dimensional distribution at a set of inducing points and then assuming conditional independence of all remaining locations given the process at those points. Such a construction can be naturally integrated with the piecewise continuous formulation of the PCGP, enabling a variant capable of approximating nonstationary GPs while preserving computational scalability and avoiding the theoretical concerns.

\section{Numerical example - spatial GP}\label{subs_NENNGP}

We present a series of numerical experiments to illustrate the practical consequences of the theoretical results regarding the probabilistic properties of existing sparse GPs and of the PCGP. 
In particular, we compare the NNGP with the PCGP using the same neighbouring structure.

We consider the bidimensional region \(\mathcal{D} = [0,10] \times [0,10]\), a reference set \(\mathcal{S}\) with \(r = 1000\) uniformly chosen locations, and a number \(m = 15\) of neighbours, following the recommendations in \cite{NNGP}. A fictitious dataset \(Y_{\mathcal{S}}\) is simulated from the NNGP approximation to a parent process with stationary mean $\mu=0$ and variance $\sigma^2=1$, and isotropic correlation function \(\rho(d) = \exp\left\{-\left(\frac{d}{\phi}\right)^{\nu}\right\}\), where \(\nu = 1.9\) and \(\phi^{\nu} = 4\). Five regular grids of locations are considered to simulate 400 replications of the process, conditional on the data \(Y_{\mathcal{S}}\), and compute various functions of it. The grids contain 625, 2500, 10,000, 40,000, and 160,000 locations each. The number $K$ of squares in the PCGP is \(16^2\) and $\ddot{m}=15$.

Figure~\ref{sim1} presents heat maps of a process realisation on a grid with 160,000 locations for both the NNGP and PCGP models (with common \(Z_{\mathcal{S}}\)). On the left-hand side of the figure, a rough transition among the color patterns is evident compared to the right-hand side, particularly along the borders of regions with higher gradients. This behavior is consistent with the theoretical results established in Theorem~\ref{RPT}. Additionally, the NNGP path is significantly noisy, highlighting its inherent irregularity.

The high path irregularity of the NNGP becomes more pronounced as the size of the reference set ($r$) or the number of neighbours ($m$) decreases. Figure~S.2 in Supplementary Material S.2 illustrates this behavior for $r = 200$ and $m = 5$. It is important to note that the datasets in the first experiment differ from those in the appendix due to the changes in $r$, $m$ and $\ddot{m}$. As $r$, $m$ and $\ddot{m}$ decrease, both the NNGP and PCGP paths exhibit more pronounced discontinuities. However, while the NNGP is everywhere discontinuous, the PCGP displays discontinuities only along the boundaries of the square regions.

Although the magnitude of the PCGP path discontinuities is influenced by the choices of $r$, $m$ and $\ddot{m}$, the path smoothness shows little variation once these parameters surpass certain thresholds. Figures~S.3 and S.4 in Supplementary Material S.2 illustrate this phenomenon, displaying simulated paths for $r = 1{,}000$, $2{,}000$, and $3{,}000$ when $m = \ddot{m} = 15$, and for $m = \ddot{m} = 15$, $30$, and $60$ when $r = 1{,}000$.
\begin{figure}[!h]
    \centering
    \includegraphics[height=.25\textheight, width=1\linewidth]{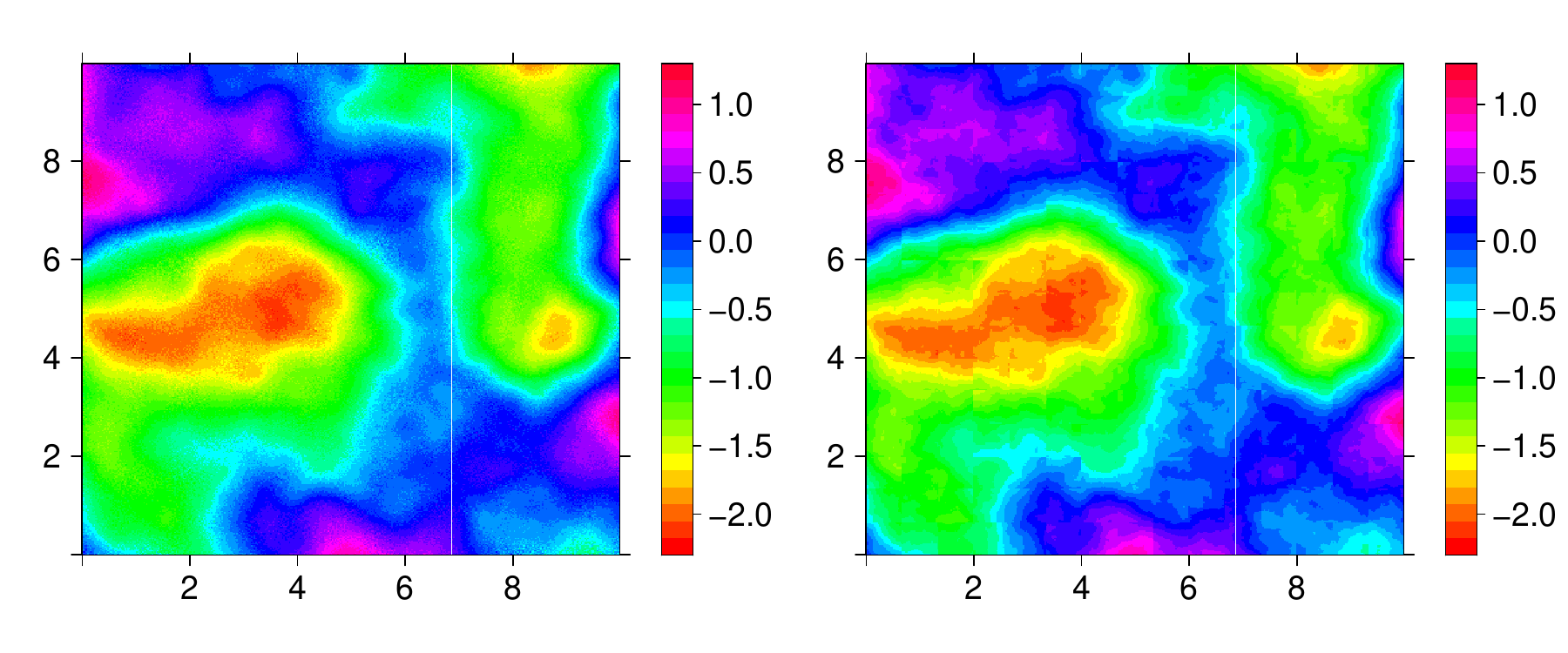}
    \caption{Heat map of one realisation of the process at a regular grid of 160,000 locations. NNGP on the left and PCGP on the right.}
    \label{sim1}
\end{figure}

We further compare the NNGP and PCGP by analyzing the following functions of the process over the regular grid:
\[
    h_1(Z) = \frac{\mu(\mathcal{D})}{M} \sum_{j=1}^{M} Z_{u_j}, \quad
    h_2(Z) = \frac{1}{M} \sum_{j=1}^{M} \mathbbm{1}_{A_i}(Z_{u_j}), \quad
    h_3(Z) = \left( \min_j Z_{u_j}, \max_j Z_{u_j} \right),
\]
where \(\mu(\mathcal{D})\) represents the area of \(\mathcal{D}\), \(M\) denotes the number of locations in the grid, and \(A_i\) (\(i = 1, 2, 3\)) are the intervals \(A_1 = (0, 0.5)\), \(A_2 = (1, 2)\), and \(A_3 = (-2.5, 0)\).

To compare the methods with the parent GP, we obtained Monte Carlo estimates by simulating replications of this on a regular grid of 6,400 points. Those estimates are considerably precise because of the smoothness of the parent GP.

Tables~\ref{tab1} and~\ref{tab2} summarise key statistics for the \(h\) functions across 400 replications for all five grids, considering the NNGP and PCGP. The results for \(h_1\) and \(h_2\) align with the theoretical findings of Theorem \ref{CPS}, as the variances of these functions are significantly smaller for the NNGP. Moreover, as \(M\) increases, these variances decrease substantially for the NNGP while stabilizing in values reasonably close to the ones of the parent GP for the PCGP. Note that, because $Z_{\mathcal{S}}$ is fixed here, Theorem \ref{CPS} implies that the variance under the NNGP converges to zero.

The results for \(h_3\) support Theorem~\ref{mMT}, which shows that the minimum and maximum of the NNGP path move towards \(-\infty\) and \(+\infty\), respectively, as \(M\) increases. This result aligns with what is observed in the right panel of Figure~\ref{fig_BB}. On the other hand, the estimates for the PCGP are fairly close to those for the parent GP. These empirical findings are even more pronounced in Tables~S.1 and~S.2 in Supplementary Material S.2, which explore scenarios with smaller values of \(r\) and \(m\).

\begin{table}[!h]
\resizebox{\columnwidth}{!}{%
\begin{tabular}{|r|rr|rr|}
\hline
\multicolumn{1}{|c|}{\textbf{M}} & \multicolumn{2}{c|}{\textbf{$h_1$}}    & \multicolumn{2}{c|}{\textbf{$h_3$}}   \\ \hline
\multicolumn{1}{|c|}{\textbf{}}  & \multicolumn{1}{c|}{\textbf{NNGP}}            & \multicolumn{1}{c|}{\textbf{PCGP}} & \multicolumn{1}{c|}{\textbf{NNGP}}            & \multicolumn{1}{c|}{\textbf{PCGP}}  \\ \hline
\textbf{625}     & \multicolumn{1}{c|}{31.99 (0.161)} & 31.95 (0.180) & \multicolumn{1}{c|}{\{-1.991,2.279\} (0.021,0.039)} & \{-1.990,2.271\} (0.023,0.040)      \\ \hline
\textbf{2,500}   & \multicolumn{1}{c|}{29.40 (0.085)} & 29.39 (0.143) & \multicolumn{1}{c|}{\{-2.016,2.309\} (0.023,0.028)} & \{-2.028,2.303\} (0.030,0.026)     \\ \hline
\textbf{10,000}  & \multicolumn{1}{c|}{28.11 (0.043)} & 28.08 (0.130) & \multicolumn{1}{c|}{\{-2.042,2.334\} (0.022,0.020)} & \{-2.041,2.319\} (0.026,0.023)     \\ \hline
\textbf{40,000}  & \multicolumn{1}{c|}{27.46 (0.021)} & 27.41 (0.135) & \multicolumn{1}{c|}{\{-2.064,2.356\} (0.019,0.020)} & \{-2.062,2.329\} (0.033,0.013)     \\ \hline
\textbf{160,000} & \multicolumn{1}{c|}{27.13 (0.010)} & 27.09 (0.129) & \multicolumn{1}{c|}{\{-2.088,2.377\} (0.021,0.017)} & \{-2.063,2.332\} (0.030,0.022)       \\ \hline
\textbf{Full GP} & \multicolumn{2}{c|}{28.40 (0.166)} & \multicolumn{2}{c|}{\{-2.034,2.316\} (0.028,0.023)}       \\ \hline
\end{tabular}}
\caption{Mean and standard deviation, in parenthesis, of statistics $h_1$ and $h_3$ over 400 replications for different grid sizes.}
\label{tab1}
\end{table}

\begin{table}[!h]
\resizebox{\columnwidth}{!}{%
\begin{tabular}{|r|rr|rr|rr|}
\hline
\multicolumn{1}{|c|}{\textbf{M}} & \multicolumn{2}{c|}{\textbf{$h_2$, $A_1$}}                                           & \multicolumn{2}{c|}{\textbf{$h_2$, $A_2$}}                                           & \multicolumn{2}{c|}{\textbf{$h_2$, $A_3$}}                                           \\ \hline
\multicolumn{1}{|c|}{\textbf{}}  & \multicolumn{1}{c|}{\textbf{NNGP}}            & \multicolumn{1}{c|}{\textbf{PCGP}} & \multicolumn{1}{c|}{\textbf{NNGP}}            & \multicolumn{1}{c|}{\textbf{PCGP}} & \multicolumn{1}{c|}{\textbf{NNGP}}            & \multicolumn{1}{c|}{\textbf{PCGP}} \\ \hline
\textbf{625}        & \multicolumn{1}{r|}{0.101 (0.004)} & \multicolumn{1}{r|}{0.102 (0.004)}      & \multicolumn{1}{r|}{0.287 (0.003)} & \multicolumn{1}{r|}{0.287 (0.004)} & \multicolumn{1}{r|}{0.344 (0.002)} & \multicolumn{1}{r|}{0.344 (0.002)}          \\ \hline
\textbf{2,500}     & \multicolumn{1}{r|}{0.098 (0.002)} & \multicolumn{1}{r|}{0.098 (0.002)}      
 & \multicolumn{1}{r|}{0.283 (0.002)} & \multicolumn{1}{r|}{0.283 (0.002)} & \multicolumn{1}{r|}{0.361 (0.001)} & \multicolumn{1}{r|}{0.361 (0.001)}         \\ \hline
\textbf{10,000}    & \multicolumn{1}{r|}{0.096 (0.001)} & \multicolumn{1}{r|}{0.097 (0.001)}       & \multicolumn{1}{r|}{0.281 (0.0009)} & \multicolumn{1}{r|}{0.281 (0.002)} & \multicolumn{1}{r|}{0.369 (0.0005)} & \multicolumn{1}{r|}{0.370 (0.0007)}        \\ \hline
\textbf{40,000}  & \multicolumn{1}{r|}{0.096 (0.0004)} & \multicolumn{1}{r|}{0.096 (0.001)}       
 & \multicolumn{1}{r|}{0.280 (0.0005)} & \multicolumn{1}{r|}{0.280 (0.002)} & \multicolumn{1}{r|}{0.373 (0.0002)} & \multicolumn{1}{r|}{0.376 (0.0007)}        \\ \hline
\textbf{160,000}  & \multicolumn{1}{r|}{0.095 (0.0002)} & \multicolumn{1}{r|}{0.096 (0.001)}       & \multicolumn{1}{r|}{0.279 (0.0002)} & \multicolumn{1}{r|}{0.279 (0.002)} & \multicolumn{1}{r|}{0.375 (0.0001)} & \multicolumn{1}{r|}{0.375 (0.0007)}         \\ \hline
\textbf{Full GP}                    & \multicolumn{2}{c|}{0.097 (0.002)}   & \multicolumn{2}{c|}{0.282 (0.002)}  & \multicolumn{2}{c|}{0.367 (0.0009)}            \\ \hline
\end{tabular}}
\caption{Mean and standard deviation, in parenthesis, of statistic $h_2$ over 400 replications for different grid sizes.}
\label{tab2}
\end{table}

We also compare the NNGP to the PCGP in terms of parameter estimation. A dataset of size 200,000 is generated, and the maximum likelihood estimates (MLEs) of the parameters are computed under both models, for $r=1000$, $m=\ddot{m} = 15$ and number of squares for the PCGP so that the number of grid points per square is around 500. The data is generated using an approximation of the full GP, where the process is simulated sequentially at each location, conditioned on the 1,000 closest previous locations, with \(\mu = 0\), \(\sigma^2 = 1\), and \(\phi^{\nu} = 4\). 

Due to theoretical limitations in jointly estimating \( \sigma^2 \) and \( \phi \) \citep[see][]{Zhang04}, we fix \( \phi \) and estimate \( \sigma^2 \) jointly with \( \mu \). Similar results are observed when fixing \( \sigma^2 \) and estimating \( \phi \). Table~\ref{tab3} summarises the results for different sample sizes. While the estimates of \( \sigma^2 \) are consistent across both models, the NNGP model exhibits mild instability in the mean (\(\mu\)) estimates. 

Additionally, given the exponential decay of the correlation function and the large grid size, the maximum likelihood estimator under the generating model approximates the sample mean, which in this case is approximately \(-0.55\) across all sample sizes.

\begin{table}[!h]
\centering
\begin{tabular}{|r|r|r|r|}
\hline
       & \multicolumn{3}{c|}{$\boldsymbol{\mu}$ / $\boldsymbol{\sigma^2}$ / \textbf{log-lik}}                         \\ \hline
\textbf{n}      & \multicolumn{1}{|c|}{$\mathbf{2\times10^3}$}     & \multicolumn{1}{|c|}{$\mathbf{5\times10^3}$}     & \multicolumn{1}{|c|}{$\mathbf{10\times10^3}$}      \\ \hline
\textbf{NNGP}   & -0.38 / 1.00 / 3.3$\times10^3$ & -0.47 / 1.00 / 9.3$\times10^3$ & -0.71 / 0.98 / 1.9$\times10^4$      \\
\textbf{PCGP} & -0.39 / 0.99 / 2.8$\times10^3$ & -0.45 / 0.99 / 1.0$\times10^4$ & -0.50 / 0.99 / 2.4$\times10^4$    \\ \hline
\textbf{n}	   & \multicolumn{1}{|c|}{$\mathbf{25\times10^3}$}    & \multicolumn{1}{|c|}{$\mathbf{50\times10^3}$}  & \multicolumn{1}{|c|}{$\mathbf{75\times10^3}$}      \\ \hline
\textbf{NNGP}   & -0.97 / 0.98 / 4.9$\times10^4$ & -1.12 / 0.98 / 9.8$\times10^4$ & -1.12 / 0.99 / 1.4$\times10^5$ \\
\textbf{PCGP} & -0.46 / 0.99 / 7.5$\times10^4$ & -0.49 / 0.98 / 1.6$\times10^5$ & -0.41 / 0.98 / 2.6$\times10^5$ \\ \hline
\textbf{n}	   & \multicolumn{1}{|c|}{$\mathbf{100\times10^3}$}   & \multicolumn{1}{|c|}{$\mathbf{150\times10^3}$} & \multicolumn{1}{|c|}{$\mathbf{200\times10^3}$}   \\ \hline
\textbf{NNGP}   & -1.07 / 0.99 / 1.9$\times10^5$ & -1.06 / 0.98 / 2.9$\times10^5$ & -1.08 / 0.98 / 3.9$\times10^5$   \\
\textbf{PCGP} & -0.41 / 1.00 / 3.6$\times10^5$ & -0.46 / 1.00 / 5.8$\times10^5$ & -0.57 / 1.01 / 8.0$\times10^5$   \\ \hline
\end{tabular}
\caption{Comparison of the maximum likelihood estimates for different data sizes. The third statistic is the log-likelihood at the MLE.}
\label{tab3}
\end{table}

Finally, we present an empirical comparison between the PCGP and the mPCGP models. Using the same specifications as previously described, Figure~\ref{sim1b} displays a simulated path, given the same \(Z_{\mathcal{S}}\), for each of the two models, with \( G = 4 \) for the mPCGP. It can be observed that the mPCGP path is noticeably smoother along the boundaries of the squares.
\begin{figure}[!h]
\centering
\includegraphics[height=.25\textheight,width=1\linewidth]{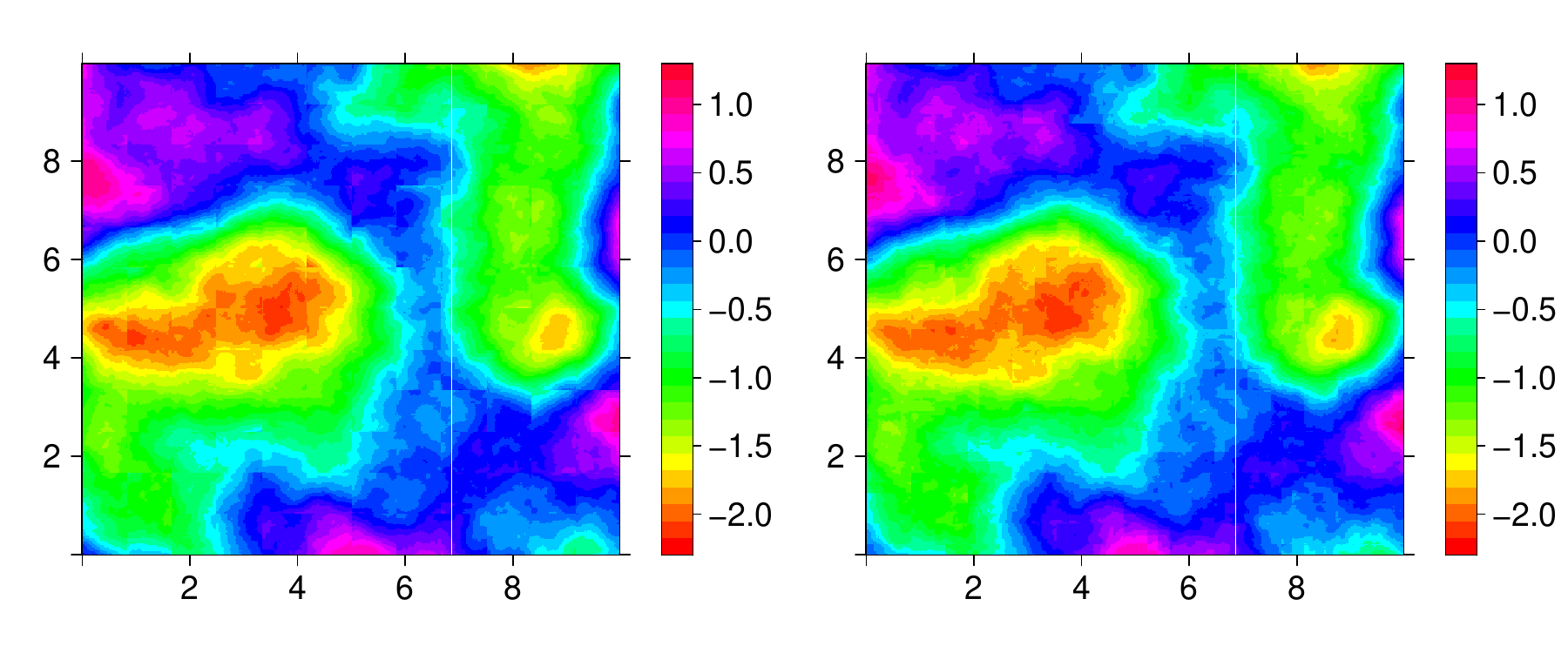}
\caption{Heat map of one realisation of the process at a regular grid of 160,000 locations using $k=12\times12$ squares. PCGP on the left and mPCGP on the right.}\label{sim1b}
\end{figure}

\section{Final remarks}
\label{s:conc}

This paper contributes two key advancements to the geostatistical literature. First, it identifies and addresses fundamental limitations within existing sparse GPs, like the NNGP and the RNGP. Specifically issues that restrict their practical application due to inherent path roughness and non-measurability of certain statistically appealing functions. Second, it introduces a novel hybrid approach that retains the computational efficiency of existing sparse GPs while preserving the smoothness and probabilistic soundness of full GPs. These contributions enable practitioners to analyse large spatial datasets more effectively, benefiting from both scalability and robust statistical properties. In particular, the PCGP can serve as a baseline Gaussian process for complex geostatistical models, including spatial fusion, quantile regression, and point process models, to perform mathematically rigorous and discretisation-error free (exact) inference.

We provide a comprehensive examination of the theoretical challenges associated with the existing sparse GPs, including path irregularities and the difficulties posed by non-measurable functions, which can compromise statistical inference and limit the utility of those processes in geostatistical applications. Our numerical experiments highlight the practical impact of these issues and demonstrate how they are effectively mitigated by the proposed PCGP.

The innovations introduced in this work lay the foundation for more reliable and scalable Gaussian processes, opening new possibilities for applications in geostatistics and beyond. This work can contribute to the future development of statistical methodologies, particularly within an exact inference framework. Furthermore, beyond spatial statistics, the machine learning community has extensively used GP's, and the challenges for well-defined, scalable spatial processes are
also relevant to this research area. Although our focus lies in the spatial statistics domain, the introduction of the PCGP, the fundamental limitations of current state-of-the-art well-defined scalable methods, and takeaways about their use can be applied and explored in a broader context within this community.

\section*{Acknowledgments}
Fl\'{a}vio Gon\c{c}alves is funded by FAPEMIG - grant APQ-01837-22, and CNPq - grant 308536/2023-1. Marcos O. Prates acknowledges the research grants obtained from CNPq-Brazil (309186/2021-8), FAPEMIG (APQ-01837-22, APQ-01748-24), and CAPES, respectively, for partial financial support.
Gareth Roberts was supported by the EPSRC grants: PINCODE (EP/X028119/1) and EP/V009478/1; and by the UKRI grant, OCEAN, EP/Y014650/1.

\bibliographystyle{apalike}
\bibliography{biblio}
\appendix
\section*{Appendix A - The probability theory behind the NNGP}\label{ProbSec}

To fully understand the issues concerning the construction of NNGPs, we shall need to delve into its precise mathematical construction, which is natural to consider through the lenses of Kolmogorov's Existence Theorem.

We follow \citet[][Sections 36 and 38]{Billin}, starting with some essential definitions. 

Let \( \mathcal{D} \subset \mathbb{R}^d \) denote a region for some \( d \in N \). A stochastic process \( X = \{X_u; u \in \mathcal{D}\} \) is a collection of random variables indexed by points in \( \mathcal{D} \). For simplicity, we assume \( X_u \) to be univariate. For any finite collection of distinct points \( (u_1, \ldots, u_k) \subset S \), the joint distribution of the random vector \( (X_{u_1}, \ldots, X_{u_k}) \) (the finite-dimensional distributions of \( X \))
is denoted by $ \mu_{u_1, \ldots, u_k}(A) = P[(X_{u_1}, \ldots, X_{u_k}) \in A], \; A \in \mathcal{B}^k.$

The process \( X \) is formally defined on the sample space \( R^\mathcal{D} \), the space of all real-valued functions on \( \mathcal{D} \). The corresponding \( \sigma \)-algebra, \( \mathcal{R}^S \), is generated by the coordinate maps \( X_u: \mathbb{R}^\mathcal{D} \to \mathbb{R} \), given by \( X_u(x) = x(u) \) for all \( u \in \mathcal{D} \). Thus, \( \mathcal{R}^\mathcal{D} \) is the smallest \( \sigma \)-field containing all sets of the form: $
[x \in \mathbb{R}^\mathcal{D} : X_u(x) \in A], \; u \in \mathcal{D}, \; A \in \mathcal{B}.
$

The collection \( \{X_u; u \in \mathcal{D}\} \) becomes a well-defined stochastic process on the measurable space \( (\mathbb{R}^\mathcal{D}, \mathcal{R}^\mathcal{D}) \) if a probability measure \( P \) is defined on this space. Kolmogorov's Existence Theorem ensures the existence of such a measure, provided the following consistency conditions hold for the finite-dimensional distributions. Let \( A = A_1 \times \ldots \times A_k \) for \( A_i \in \mathcal{B} \), and let \( \pi \) be a permutation of \( (1, \ldots, k) \). The conditions are:
\begin{align}
\mu_{u_1, \ldots, u_k}(A_1 \times \ldots \times A_k) &= \mu_{u_{\pi 1}, \ldots, u_{\pi k}}(A_{\pi 1} \times \ldots \times A_{\pi k}), \label{CC1} \\
\mu_{u_1, \ldots, u_{k-1}}(A_1 \times \ldots \times A_{k-1}) &= \mu_{u_1, \ldots, u_k}(A_1 \times \ldots \times A_{k-1}, \mathbb{R}). \label{CC2}
\end{align}

\begin{theorem} 
\label{KET}
\textbf{Kolmogorov's Existence Theorem} \citep{Billin}.  
If $\mu_{u_1,\ldots,u_k}$ are a system of distributions satisfying the consistence conditions \eqref{CC1} and \eqref{CC2}, then there is a probability measure $P$ on $\mathcal{R}^\mathcal{D}$ such that the process $\{X_u;u\in \mathcal{D}\}$ on $(\mathbb{R}^\mathcal{D},\mathcal{R}^\mathcal{D},P)$ has the $\mu_{u_1,\ldots,u_k}$ as its finite-dimensional distributions.
\end{theorem}

The critical observation from Kolmogorov's theorem is that the process is guaranteed to exist only on the measurable space \( (R^\mathcal{D}, \mathcal{R}^\mathcal{D}) \). However, several important sets are not included in this \( \sigma \)-field. For example, the set of continuous functions on \( \mathcal{D} \) lies outside \( \mathcal{R}^\mathcal{D} \) \citep[see][Section 36]{Billin}. In general, a set \( A \in \mathcal{R}^\mathcal{D} \) must satisfy the following property: there exists a countable subset \( T \subset \mathcal{D} \) such that if \( x \in A \) and \( x(u) = y(u) \) for all \( u \in T \), then \( y \in A \). This implies that only functions determined by a countable collection of locations are measurable, which excludes many relevant functions, such as:
$$
\int_{\mathcal{D}} g(X_u) \, du \quad \text{and} \quad C(X) = \{u \in \mathcal{D} : X_u \in B\}, \; B \in \mathcal{B}.
$$

As \citet{Billin} remarked: ``An application of Kolmogorov's Existence Theorem will always yield a stochastic process with the prescribed finite-dimensional distributions, but the process may lack certain path-function properties that are reasonable to require of it as a model for some natural phenomenon."

We consider processes defined by independent random variables, reflecting the conditional independence that characterizes the infinite-dimensional remainder of the NNGP given the finite reference set. Specifically, consider \( X = \{X_u; u \in \mathcal{D}\} \) where each \( X_u \) is a continuous random variable with common support \( C \subset R \). While Kolmogorov's theorem ensures the existence of such a process, its paths are almost surely highly irregular. As noted by \citet[][Example 38.5]{Billin}, the process is almost surely everywhere dense on \( C \), implying for that: $
\inf_{u \in \mathcal{D}} X_u \; \text{and} \; \sup_{u \in \mathcal{D}} X_u
$ 
are simply the lower and upper limits of \( C \). Additionally, functions such as \( C(X) = \{u \in \mathcal{D} : X_u \in B\} \) are not only non-measurable but also exhibit highly irregular discontinuous behavior, further limiting their usefulness.

\section*{Appendix B - Proofs}

\begin{proof}[Proof of Theorem \ref{RPT}]
Fix $z \in \mathbb R^r$, and let $P_z^*(\cdot)=P(\cdot\mid Z_{\mathcal S}=z)$. Under $P_z^*$, the variables $\{Z_u:u\in \mathcal D\setminus \mathcal S\}$ are independent and Gaussian, with conditional mean $\mu_z(u)$ and conditional variance $\sigma^2(u)$.
Let $\mathfrak D$ denote the countable collection of open hypercubes contained in $\mathcal D$ with rational endpoints, and let $\mathfrak A$ denote the countable collection of open intervals in $\mathbb R$ with rational endpoints. Fix $(\mathcal D^*,A)\in \mathfrak D\times \mathfrak A$. Since $\mathcal S$ is finite and $\mathcal D^*$ is open, one can choose $\varepsilon_{\mathcal D^*}>0$ such that \;
\(\
K_{\mathcal D^*}:=\mathcal D^* \setminus \bigcup_{s\in\mathcal S}\overline B(s,\varepsilon_{\mathcal D^*})
\) \;
is a nonempty open subset of $\mathcal D^*$. In particular, $K_{\mathcal D^*}$ lies at positive distance from $\mathcal S$.

By continuity of the conditional variance away from $\mathcal S$, there exist constants $0<\underline\sigma_{\mathcal D^*}\le \overline\sigma_{\mathcal D^*}<\infty$ such that $\underline\sigma_{\mathcal D^*}^2\le \sigma^2(u)\le \overline\sigma_{\mathcal D^*}^2$ for all $u\in K_{\mathcal D^*}$. Moreover, for fixed $z$, continuity of $u\mapsto \mu_z(u)$ implies that $|\mu_z(u)|\le M_{z,\mathcal D^*}$ for all $u\in K_{\mathcal D^*}$, for some finite constant $M_{z,\mathcal D^*}$. Hence
\(
p_{z,\mathcal D^*,A}:=\inf_{u\in K_{\mathcal D^*}} P_z^*(Z_u\in A)>0
\).

Now $D_0\cap K_{\mathcal D^*}$ is infinite, since $D_0$ is dense in $\mathcal D\setminus \mathcal S$ and $K_{\mathcal D^*}$ is a nonempty open subset of $\mathcal D\setminus \mathcal S$. Enumerate it as $\{u_n\}_{n\ge1}$, and define $E_n=\{Z_{u_n}\in A\}$. The events $E_n$ are independent under $P_z^*$ and satisfy $P_z^*(E_n)\ge p_{z,\mathcal D^*,A}$ for all $n$. Therefore $\sum_{n=1}^\infty P_z^*(E_n)=\infty$, and the second Borel-Cantelli lemma yields \;
\(
P_z^*(E_n \text{ i.o.})=1.
\) \;
Thus, for every $(\mathcal D^*,A)\in \mathfrak D\times \mathfrak A$,
\mbox{\(
P_z^*\bigl(\exists\ \text{infinitely many }u\in D_0\cap \mathcal D^* \text{ such that } Z_u\in A\bigr)=1\).}

Since $\mathfrak D\times \mathfrak A$ is countable, the intersection of these events over all $(\mathcal D^*,A)\in \mathfrak D\times \mathfrak A$ still has $P_z^*$-probability one. On that event, for every $\mathcal D^*\in\mathfrak D$, the set $\{Z_u:u\in D_0\cap \mathcal D^*\}$ intersects every open interval with rational endpoints, and is therefore dense in $\mathbb R$. Hence, $
P_z^*\bigl(\{Z_u:u\in D_0\cap \mathcal D^*\}\text{ is dense in }\mathbb R \text{ for every } \mathcal D^*\in\mathfrak D\bigr)=1.
$ Finally, let \(I\) denote the indicator of this. Then
\(P(I=1)=E\bigl[P(I=1\mid Z_{\mathcal S})\bigr]\allowbreak=E[1]=1\).
\end{proof}

\begin{proof}[Proof of Theorem \ref{mMT}]
Let $D_\infty=\bigcup_{n\geq 1}D_n$. Since the grids are regular and their precision is successively doubled, $D_\infty$ is countable and dense in $\mathcal D$. Hence, for any open hypercube $\mathcal D^* \subset \mathcal D$ with rational endpoints, $D_\infty\cap \mathcal D^*\cap (\mathcal D\setminus\mathcal S)$ is countable and dense in $\mathcal D^*$. By Theorem~\ref{RPT}, $\{Z_u: u\in D_\infty\cap \mathcal D^*\cap (\mathcal D\setminus\mathcal S)\}$ 
is $P$-a.s.\ dense in $\mathbb{R}$.
Thus, for any $x,y\in\mathbb{R}$, there exist infinitely many $u\in D_\infty$ such that $Z_u<x$ and infinitely many $u\in D_\infty$ such that $Z_u>y$, proving (i).

Now fix $M>0$. By (i), $P$-a.s.\ there exist $u,v\in D_\infty$ such that $Z_u<-M$ and $Z_v>M$. Since $D_n\uparrow D_\infty$, both $u$ and $v$ belong to $D_n$ for all sufficiently large $n$, so $
\min_{w\in D_n} Z_w<-M,
\;
\max_{w\in D_n} Z_w>M$ 
for all large enough $n$. As $M>0$ is arbitrary, (ii) follows.
\end{proof}

\begin{proof}[Proof of Theorem~\ref{CPS}]
Let $W_i\sim N(0,1)$ be i.i.d. and write $Z_{u_i}=\sum_{j=1}^r a_{ij} Z_{\mathfrak{s}_j} + \sigma_i W_i$.

\textbf{Part (i).}
We have
\[
\sum_{i=1}^{r_n} \frac{1}{r_n} Z_{u_i}
=
\sum_{j=1}^{r}
\left(
\sum_{i=1}^{r_n} \frac{a_{ij}}{r_n}
\right) Z_{\mathfrak{s}_j}
+
\sum_{i=1}^{r_n} \frac{\sigma_i W_i}{r_n}.
\]
Since $\sigma_i \le \sigma$ for some $\sigma\in\mathbb R^+$, the second term converges to $0$ $P$-a.s.\ by \citet[Theorem~2.9.2]{revesz}.

For each $j=1,\ldots,r$, the coefficient $a_{ij}$ may be written as $a_j(u_i)$, where $a_j:\mathcal D\to\mathbb R$ is induced by the conditional Gaussian mean under the $m$-nearest-neighbor construction. Since $\mathcal S$ is finite, there are only finitely many possible neighbor sets, inducing a finite partition of $\mathcal D$. On each region of this partition, the neighbor set is fixed, and $a_j$ is continuous by continuity of the covariance function. The boundaries between regions are contained in finitely many distance-tie hypersurfaces and hence have Lebesgue measure zero. Therefore $a_j$ is Riemann integrable on $\mathcal D$, and the regularity of the grids implies that $
\lim_{n\to\infty}\sum_{i=1}^{r_n}\frac{a_{ij}}{r_n}
$ exists for each $j$. The result follows.

\medskip
\textbf{Part (ii).}
Let $\mu_i=\sum_{j=1}^r a_{ij} Z_{\mathfrak{s}_j}$. Then
\[
\sum_{i=1}^{r_n}\frac{1}{r_n}g(Z_{u_i})
=
\frac{\sum_{i=1}^{r_n}\bigl(g(\mu_i+\sigma_iW_i)-\mu_{g,i}\bigr)}{r_n}
+
\frac{\sum_{i=1}^{r_n}\mu_{g,i}}{r_n}.
\]
Since $\mu_i$ is bounded and $g(\mu+\sigma W)$ has finite variance for all $(\mu,\sigma)\in\mathbb R\times\mathbb R^+$, the first term converges to $0$ $P$-a.s.\ by \citet[Theorem~2.9.2]{revesz}.

Moreover, $\mu_i=\mu(u_i)$ and $\sigma_i^2=\sigma^2(u_i)$ for piecewise continuous functions $\mu,\sigma^2$ on the same finite partition of $\mathcal D$. Hence, $ \mu_{g,i}=h(u_i),
\;
h(u):=\mathbb E_W[g(\mu(u)+\sigma(u)W)],$ 
where $h$ is piecewise continuous, with discontinuities confined to a set of Lebesgue measure zero. Therefore $h$ is Riemann integrable on $\mathcal D$, and the regularity of the grids implies that $\lim_{n\to\infty}\sum_{i=1}^{r_n}\frac{\mu_{g,i}}{r_n}$
exists. This proves the result.
\end{proof}

\begin{proof}[Proof of Theorem \ref{Exst_PCGP}]

The PCGP is specified through (i) a marginal distribution for $Z_{\mathcal S}$ and (ii) a conditional distribution for the regional restrictions $\{Z_{\mathcal{D}_k}\}_{k=1}^K$ given $Z_{\mathcal S}$.
We show that these specifications determine a unique probability measure via a standard kernel construction.

Let $(\mathbb R^r,\mathcal B^r,P_1)$ denote the probability space of $Z_{\mathcal S}$ defined by \eqref{eq:pcgp-reference}.
For each $k\in\{1,\ldots,K\}$, let $C_k := C(\mathcal{D}_k)$ be the Banach space of real-valued continuous functions on $\mathcal{D}_k$, endowed with the supremum norm and its associated Borel $\sigma$-field $\mathcal C_k$.
Let $P$ denote the probability measure of the parent Gaussian process on $C(\mathcal{D})$, induced by a positive definite covariance function and continuous sample paths.

For each $z\in\mathbb R^r$, define $P_{2,k,z}$ as the conditional law of the restriction $Z_{\mathcal{D}_k}$ under the parent Gaussian process given $Z_{\mathcal S}=z$ (or, more generally, given the subset of reference locations used to condition region $\mathcal{D}_k$).
Since the parent Gaussian process induces a Borel probability measure on the Banach space $C(\mathcal{D})$, conditioning on finitely many evaluation functionals yields regular conditional distributions.
Consequently, for each $k$, the mapping $z \mapsto P_{2,k,z}$
defines a probability kernel from $(\mathbb R^r,\mathcal B^r)$ to $(C_k,\mathcal C_k)$.
Equivalently, for every $A_k\in\mathcal C_k$, the function $z\mapsto P_{2,k,z}(A_k)$ is $\mathcal B^r$-measurable.

\medskip
\noindent\emph{Step 1: product kernel.}
For each $z\in\mathbb R^r$, define the product measure
\[
P_{2,z} := \bigotimes_{k=1}^K P_{2,k,z}
\;\text{on}\;
\Big(\bigotimes_{k=1}^K C_k,\ \bigotimes_{k=1}^K \mathcal C_k\Big).
\]
For a measurable rectangle $A=\prod_{k=1}^K A_k$ with $A_k\in\mathcal C_k$, we have
$P_{2,z}(A) = \prod_{k=1}^K P_{2,k,z}(A_k),
$ which is $\mathcal B^r$-measurable in $z$ by the kernel property and closure of measurable functions under products.
Let
\(
\mathcal F := \Big\{A\in\bigotimes_{k=1}^K\mathcal C_k:\ z\mapsto P_{2,z}(A)\ \text{is $\mathcal B^r$-measurable}\Big\}
\).
Then $\mathcal F$ is a Dynkin system (equivalently, a monotone class) containing the measurable rectangles, and hence it contains the $\sigma$-algebra they generate, i.e.\ $\mathcal F=\bigotimes_{k=1}^K\mathcal C_k$.
Therefore, $z\mapsto P_{2,z}(A)$ is $\mathcal B^r$-measurable for all $A\in\bigotimes_{k=1}^K\mathcal C_k$.

\medskip
\noindent\emph{Step 2a: joint measure on $\mathbb R^r\times \otimes_k C_k$.}
Define $\tilde P$ on the product space
\[
\Big(\mathbb R^r\times \bigotimes_{k=1}^K C_k,\ \mathcal B^r\otimes \bigotimes_{k=1}^K\mathcal C_k\Big)
\mbox{ by }
\tilde P(B\times A)
:= \int_B P_{2,z}(A)\,P_1(dz),
\;
B\in\mathcal B^r,\ A\in\bigotimes_{k=1}^K\mathcal C_k.
\]
By the measurability established above, $\tilde P$ is well defined and extends uniquely to a probability measure on the full product $\sigma$-algebra (e.g., by Lemma~1 of \citet{POGAMP}).

\medskip
\noindent\emph{Step 2b: marginal measure on $\otimes_k C_k$ only.}
Alternatively, define the PCGP directly as a probability measure on
$\big(\bigotimes_{k=1}^K C_k,\ \bigotimes_{k=1}^K\mathcal C_k\big)$ via the mixture
\(
\tilde P(A) := \int_{\mathbb R^r} P_{2,z}(A)\,P_1(dz)\),
 \(A\in\bigotimes_{k=1}^K\mathcal C_k
\).
This yields the law of the random element $(Z_{\mathcal{D}_1},\ldots,Z_{\mathcal{D}_K})$ and makes the explicit $\mathbb R^r$ component redundant.
\end{proof}

\newpage

\setcounter{subsection}{0}
\setcounter{section}{0}
\renewcommand{\thesection}{S.\arabic{section}}
\renewcommand{\thesubsection}{S.\arabic{section}.\arabic{subsection}}
\renewcommand\theequation{S.\arabic{equation}}
\renewcommand\thetable{S.\arabic{table}}
\renewcommand\thefigure{S.\arabic{figure}}

\begin{center}
\LARGE{\textbf{Supplementary Material}}
\end{center}

\section{More on the mPCGP}
\label{sec:mPCGP}

Figure~\ref{sim5} illustrates examples of partitions with $G = 2$ and $G = 4$, using squares in $R^2$.

\captionsetup[figure]{skip=0pt}
\begin{figure}[!h]
\centering
\includegraphics[width=0.7\linewidth]{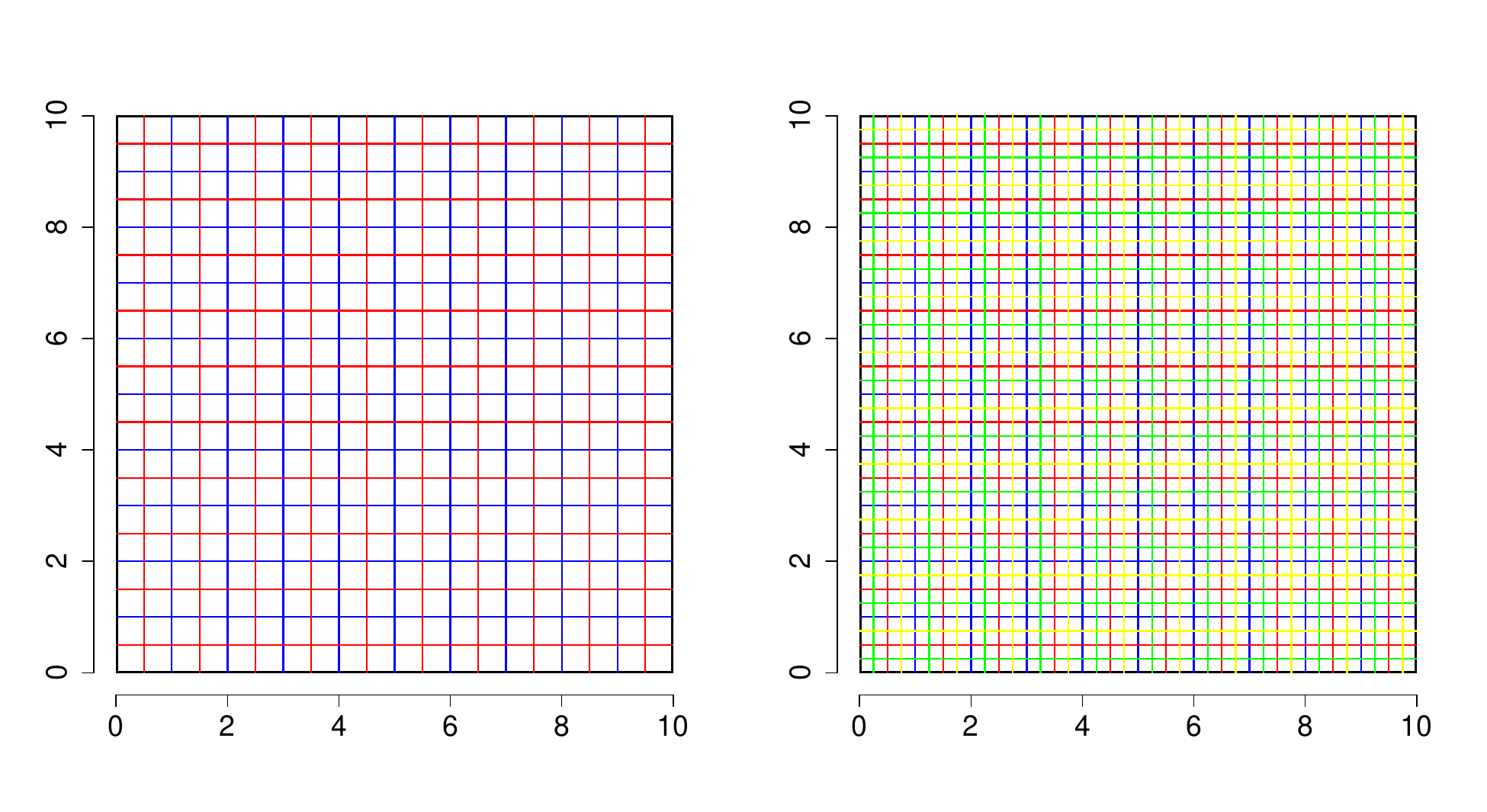}
\caption{Examples of grids for the mPCGP with $G=2$ (left) and $G=4$ (right). Each colour represents one partition.}\label{sim5}
\end{figure}

\section{Further numerical results}
\label{sec:num}

This section presents futher results of the numerical experiment. Figure~\ref{sim2} illustrates this behavior for $r=200$ and $m = \ddot{m} = 5$. Figures~\ref{sim3} and \ref{sim4} show variation of $r$ and $m$ for the PCGP. Finally, Tables~\ref{tab1b} and \ref{tab2b} summarise key statistics for the $h$ functions across $400$ replications for all five grids, considering the NNGP and PCGP with $r=50$ and $m = \ddot{m} = 5$.

\captionsetup[figure]{skip=0pt}
\begin{figure}[!h]
\centering
\includegraphics[width=1\linewidth]{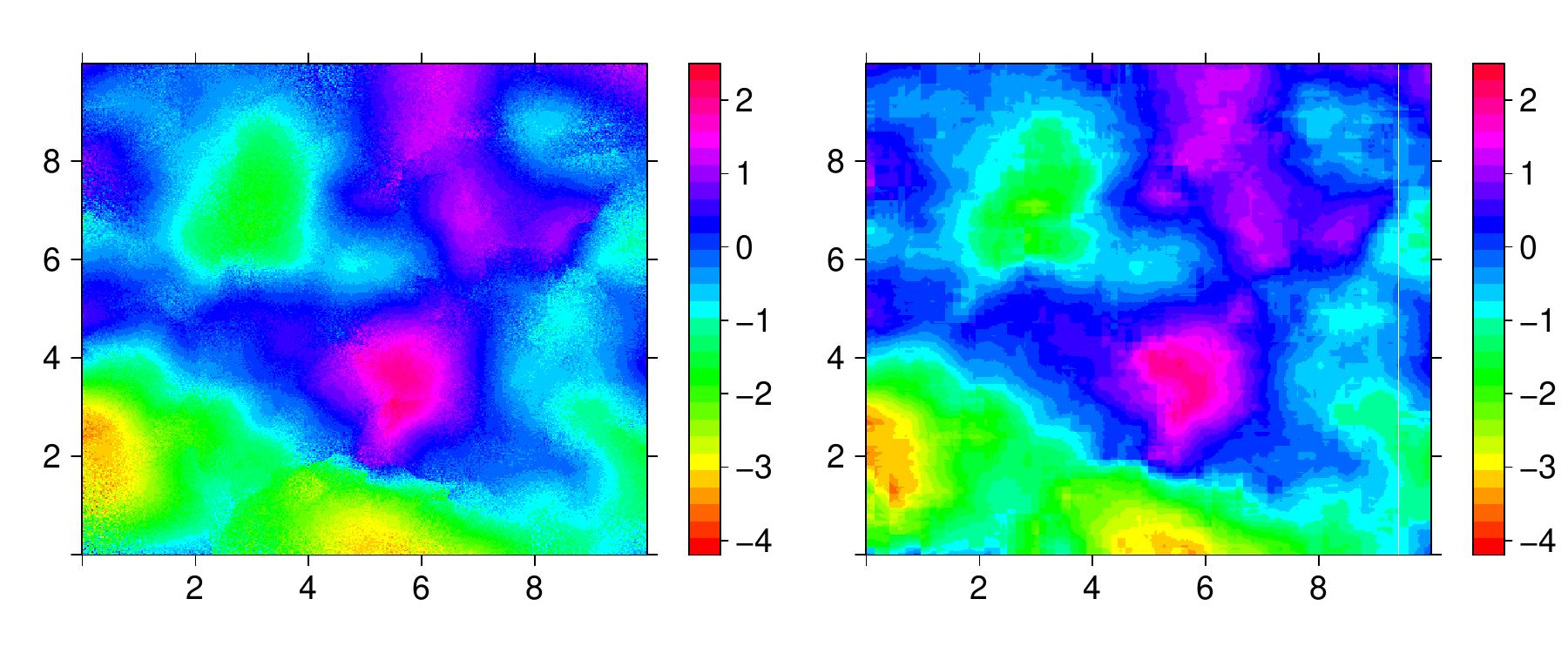}
\caption{Heat map of one realisation of the process at a regular grid of 160,000 locations for values $r=200$ and $m=5$. NNGP on the left and PCGP on the right.}\label{sim2}
\end{figure}

\captionsetup[figure]{skip=0pt}
\begin{figure}[!h]
\centering
\includegraphics[width=1\linewidth]{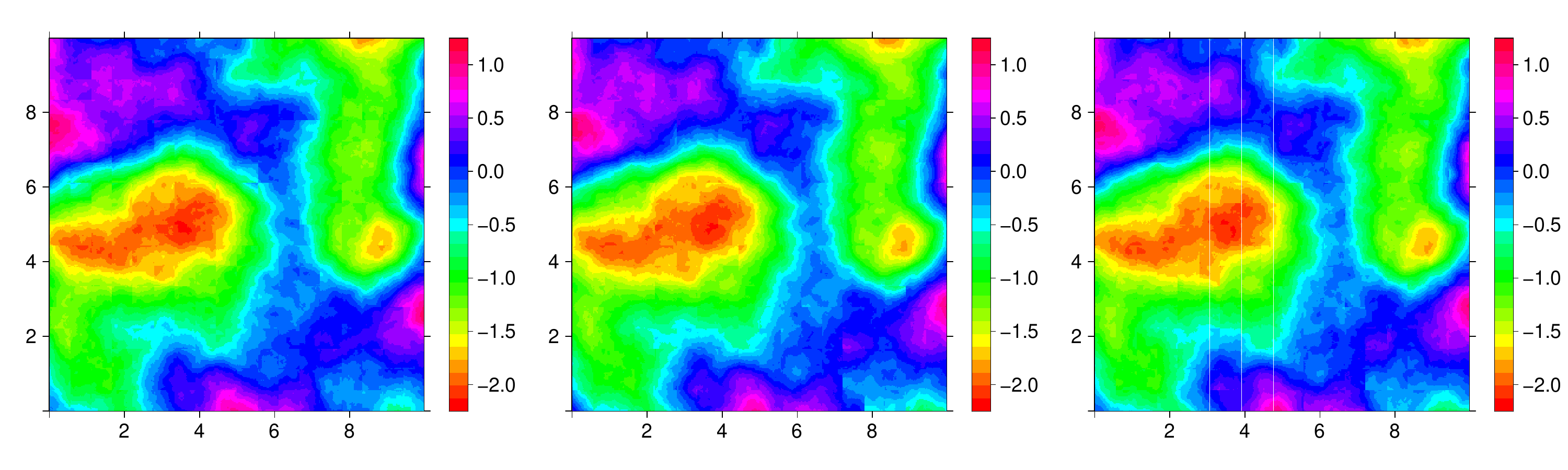}
\caption{Heat map of one realisation of the process at a regular grid of 160,000 locations for values $r=1000$, $2000$ and $3000$, with $m=15$, for the PCGP.}\label{sim3}
\end{figure}

\captionsetup[figure]{skip=0pt}
\begin{figure}[!h]
\centering
\includegraphics[width=1\linewidth]{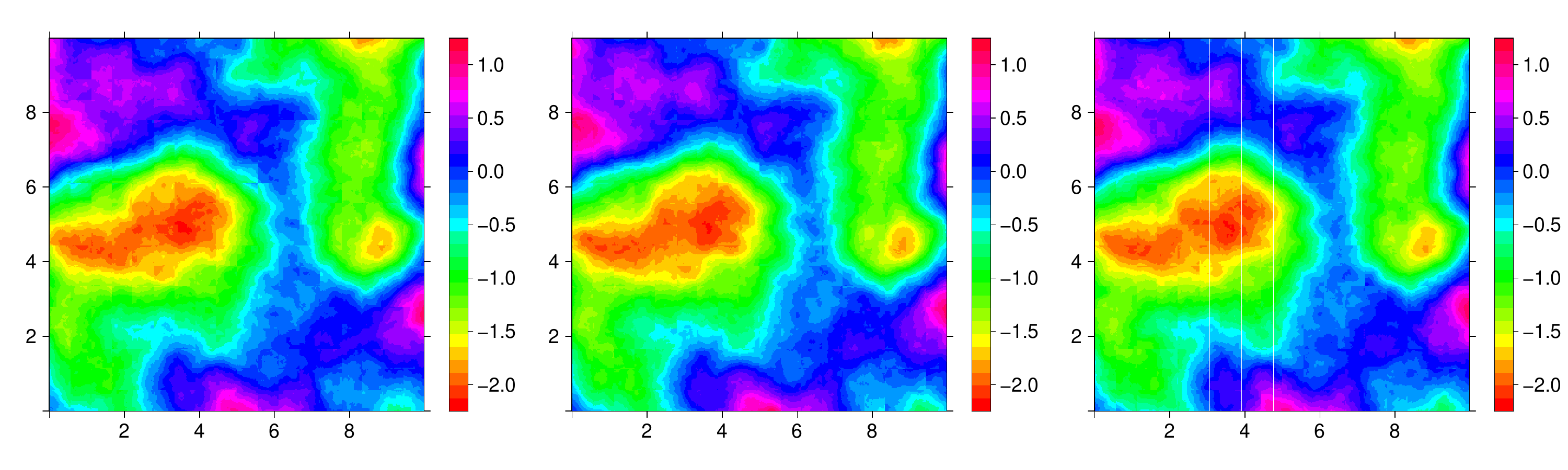}
\caption{Heat map of one realisation of the process at a regular grid of 160,000 locations for values $m=15$, $30$ and $60$, with $r=1000$, for the PCGP.}\label{sim4}
\end{figure}

\begin{table}[!h]
\resizebox{\columnwidth}{!}{%
\begin{tabular}{|r|rr|rr|}
\hline
\multicolumn{1}{|c|}{\textbf{M}} & \multicolumn{2}{c|}{\textbf{$h_1$}}    & \multicolumn{2}{c|}{\textbf{$h_3$}}   \\ \hline
\multicolumn{1}{|c|}{\textbf{}}  & \multicolumn{1}{c|}{\textbf{NNGP}}            & \multicolumn{1}{c|}{\textbf{PCGP}} & \multicolumn{1}{c|}{\textbf{NNGP}}            & \multicolumn{1}{c|}{\textbf{PCGP}}  \\ \hline
\textbf{625}            & \multicolumn{1}{r|}{-14.96 (1.431)} & \multicolumn{1}{r|}{-14.82 (2.510)} & \multicolumn{1}{r|}{\{-3.260,2.558\} (0.181,0.333)} & \multicolumn{1}{r|}{\{-3.230,2.410\} (0.183,0.375)}      \\ \hline
\textbf{2,500}          & \multicolumn{1}{r|}{-14.04 (0.767)} & \multicolumn{1}{r|}{-13.96 (2.308)} & \multicolumn{1}{r|}{\{-3.434,2.884\} (0.179,0.336)} & \multicolumn{1}{r|}{\{-3.276,2.477\} (0.182,0.367)}      \\ \hline
\textbf{10,000}         & \multicolumn{1}{r|}{-13.53 (0.349)} & \multicolumn{1}{r|}{-13.55 (2.198)} & \multicolumn{1}{r|}{\{-3.642,3.210\} (0.196,0.303)} & \multicolumn{1}{r|}{\{-3.313,2.542\} (0.185,0.378)}      \\ \hline
\textbf{40,000}         & \multicolumn{1}{r|}{-13.31 (0.178)} & \multicolumn{1}{r|}{-13.33 (2.210)} & \multicolumn{1}{r|}{\{-3.838,3.520\} (0.189,0.271)} & \multicolumn{1}{r|}{\{-3.317,2.589\} (0.178,0.367)}      \\ \hline
\textbf{160,000}        & \multicolumn{1}{r|}{-13.20 (0.086)} & \multicolumn{1}{r|}{-13.26 (2.167)} & \multicolumn{1}{r|}{\{-4.043,3.811\} (0.188,0.283)} & \multicolumn{1}{r|}{\{-3.331,2.562\} (0.181,0.371)}       \\ \hline
\textbf{Full GP}        & \multicolumn{2}{c|}{-10.17 (4.961)} & \multicolumn{2}{c|}{\{-3.192,2.281\} (0.147,0.413)}       \\ \hline
\end{tabular}}
\caption{Mean and standard deviation, in parenthesis, of statistics $h_1$ and $h_3$ over 400 replications for different grid sizes, $r=50$ and $m=5$.}
\label{tab1b}
\end{table}

\begin{table}[!h]
\resizebox{\columnwidth}{!}{%
\begin{tabular}{|r|rr|rr|rr|}
\hline
\multicolumn{1}{|c|}{\textbf{M}} & \multicolumn{2}{c|}{\textbf{$h_2$, $A_1$}}                                           & \multicolumn{2}{c|}{\textbf{$h_2$, $A_2$}}                                           & \multicolumn{2}{c|}{\textbf{$h_2$, $A_3$}}                                           \\ \hline
\multicolumn{1}{|c|}{\textbf{}}  & \multicolumn{1}{c|}{\textbf{NNGP}}            & \multicolumn{1}{c|}{\textbf{PCGP}} & \multicolumn{1}{c|}{\textbf{NNGP}}            & \multicolumn{1}{c|}{\textbf{PCGP}} & \multicolumn{1}{c|}{\textbf{NNGP}}            & \multicolumn{1}{c|}{\textbf{PCGP}} \\ \hline
\textbf{625}         & \multicolumn{1}{r|}{0.222 (0.013)} & \multicolumn{1}{r|}{0.220 (0.016)}     & \multicolumn{1}{r|}{0.091 (0.008)} & \multicolumn{1}{r|}{0.091 (0.012)} & \multicolumn{1}{r|}{0.518 (0.012)} & \multicolumn{1}{r|}{0.517 (0.017)}     \\ \hline
\textbf{2,500}       & \multicolumn{1}{r|}{0.222 (0.006)} & \multicolumn{1}{r|}{0.220 (0.093)}     & \multicolumn{1}{r|}{0.092 (0.004)} & \multicolumn{1}{r|}{0.093 (0.011)} & \multicolumn{1}{r|}{0.517 (0.006)} & \multicolumn{1}{r|}{0.517 (0.015)}     \\ \hline
\textbf{10,000}      & \multicolumn{1}{r|}{0.222 (0.003)} & \multicolumn{1}{r|}{0.219 (0.014)}     & \multicolumn{1}{r|}{0.094 (0.002)} & \multicolumn{1}{r|}{0.094 (0.010)} & \multicolumn{1}{r|}{0.517 (0.003)} & \multicolumn{1}{r|}{0.517 (0.015)}     \\ \hline
\textbf{40,000}      & \multicolumn{1}{r|}{0.222 (0.001)} & \multicolumn{1}{r|}{0.221 (0.013)}     & \multicolumn{1}{r|}{0.094 (0.001)} & \multicolumn{1}{r|}{0.095 (0.010)} & \multicolumn{1}{r|}{0.517 (0.001)} & \multicolumn{1}{r|}{0.517 (0.014)}    \\ \hline
\textbf{160,000}     & \multicolumn{1}{r|}{0.222 (0.0008)} & \multicolumn{1}{r|}{0.218 (0.012)}    & \multicolumn{1}{r|}{0.094 (0.0005)} & \multicolumn{1}{r|}{0.095 (0.010)} & \multicolumn{1}{r|}{0.517 (0.0008)} & \multicolumn{1}{r|}{0.519 (0.014)}    \\ \hline
\textbf{Full GP}     & \multicolumn{2}{c|}{0.233 (0.022)}  & \multicolumn{2}{c|}{0.100 (0.020)}  & \multicolumn{2}{c|}{0.502 (0.030)}            \\ \hline
\end{tabular}}
\caption{Mean and standard deviation, in parenthesis, of statistic $h_2$ over 400 replications for different grid sizes, $r=50$ and $m=5$.}
\label{tab2b}
\end{table}

%\bibliographystyle{apalike}%chicago, harvard
%\bibliography{biblio}
\end{document}